\documentclass[twoside,leqno,twocolumn]{article}  
\usepackage{ltexpprt}
\usepackage{amsfonts}
\usepackage{amssymb,amsmath}
\usepackage{epsfig}
\usepackage{color}
\usepackage{latexsym}
\usepackage{enumerate}
\usepackage{algorithm}
\usepackage{algorithmic}

\usepackage{tweaklist}
\renewcommand{\enumhook}{\setlength{\topsep}{2pt}%
  \setlength{\itemsep}{-2pt}}
\renewcommand{\itemhook}{\setlength{\topsep}{2pt}%
  \setlength{\itemsep}{-2pt}}
 \newtheorem{definition}{Definition}

\newtheorem{observation}{Observation}

\newtheorem{remark}{Remark}

\newtheorem{informaltheorem}{Informal Theorem}
\newenvironment{prevproof}[2]{\noindent {\em {Proof of {#1}~\ref{#2}:}}}{$\Box$\vskip \belowdisplayskip}

\newcommand{\junk}[1]{}
\junk{

}

\newcommand{\poly}{\text{poly}}
\newcommand{\opt}{\text{OPT}}

\newcommand{\mattnote}[1]{{\color{black}{#1}}}

\newcommand{\yangnote}[1]{{\color{black}{#1}}}
\newcommand{\notshow}[1]{{}}

\DeclareMathOperator{\argmax}{argmax}

\definecolor{MyGray}{rgb}{0.8,0.8,0.8}

\begin{document}
\title{{Reducing Revenue to Welfare Maximization: Approximation Algorithms and other Generalizations}}
\author {Yang Cai\thanks{Supported by NSF Award CCF-0953960 (CAREER) and CCF-1101491.}\\
EECS, MIT \\
\tt{ycai@csail.mit.edu}
\and
Constantinos Daskalakis\thanks{Supported by a Sloan Foundation Fellowship, a Microsoft Research Faculty Fellowship and NSF Award CCF-0953960 (CAREER) and CCF-1101491.}\\
EECS, MIT \\
\tt{costis@mit.edu}
\and
S. Matthew Weinberg\thanks{Supported by a NSF Graduate Research Fellowship and NSF award CCF-1101491.}\\
EECS, MIT\\
\tt{smw79@mit.edu}
}
\addtocounter{page}{-1}
\date{}
\maketitle
\begin{abstract}

It was recently shown in~\cite{CaiDW12b} that revenue optimization can be computationally efficiently reduced to welfare optimization in all multi-dimensional Bayesian auction problems with arbitrary (possibly combinatorial) feasibility constraints and independent additive bidders with arbitrary (possibly combinatorial) demand constraints. This reduction provides a poly-time solution to the optimal mechanism design problem in all auction settings where welfare optimization can be solved efficiently, but it is fragile to approximation and cannot provide solutions to settings where welfare maximization can only be tractably approximated. In this paper, we extend the reduction to accommodate approximation algorithms, providing an approximation preserving reduction from (truthful) revenue maximization to (not necessarily truthful) welfare maximization. The mechanisms output by our reduction choose allocations via black-box calls to welfare approximation on randomly selected inputs, thereby generalizing also our earlier structural results on optimal multi-dimensional mechanisms to approximately optimal mechanisms. Unlike~\cite{CaiDW12b}, our results here are obtained through novel uses of the Ellipsoid algorithm and other optimization techniques over {\em non-convex regions}. 

%
%
\end{abstract}
\thispagestyle{empty}


\section{Introduction} \label{sec:intro}

The \emph{optimal mechanism design} problem is a central problem in mathematical economics that has recently  gained a lot of attention in computer science. The setting for this problem is simple: a seller has a limited supply of several items for sale and many buyers interested in these items. The problem is to design an auction for the buyers to play that will maximize the seller's revenue. While being able to solve this problem in the worst-case would be desirable, it is easy to see that there can't be any meaningful worst-case solution. Indeed, even in the simpler case of a single buyer and a single item, how would the seller sell the item to optimize her profit without any assumptions about the buyer who can, in principle, lie about his willingness to pay? 

To cope with this impossibility, economists have taken a Bayesian approach, assuming that a prior distribution is known that determines the values of the buyers for each item (and each bundle of items), and aiming to optimize the seller's revenue in expectation over bidders sampled from this distribution. Under this assumption, Myerson solved the single-item version of the  problem in a seminal paper on Bayesian mechanism design~\cite{Myerson81}. On the other hand, until very recently there had been very small progress on the ``multi-dimensional'' version of the problem, i.e. the setting where the seller has multiple {\em heterogeneous} items for sale~\cite{ManelliV07}. This challenging, and more general problem is the focus of this paper. 

We proceed to state the problem we study more precisely. First, as we move beyond single-item settings, one may want to consider settings with non-trivial \emph{feasibility constraints} on which bidders may simultaneously receive which items. Here are some examples:
\begin{enumerate}
\item Maybe the items are baseball cards. Here, a feasible allocation should award each card to at most one bidder.
\item Maybe the items are houses. Here, a feasible allocation should award each house to at most one bidder, and to each bidder at most one house.
\item Maybe the items are links in a network, and all bidders have a desired source-destination pair. Here, a feasible allocation should award each link to at most one bidder, and to each bidder a simple path from their source to their destination (or nothing).
\end{enumerate}
Following the notation of~\cite{CaiDW12b} we encapsulate the feasibility constraints that the auctioneer faces in a set system $\mathcal{F}$ on possible allocations. Namely, if the bidder set is $[m]$ and the item set is $[n]$ then an allocation of items to bidders can be represented by a subset $O \subseteq [m]\times[n]$.  ${\cal F}$ is then a set-system ${\cal F} \subseteq 2^{[m] \times [n]}$, determining what allocations are allowed for the auctioneer to choose from. 

With this notation in place, we formally state our problem, and note that virtually every recent result in the revenue-maximization literature~\cite{AlaeiFHHM12,BhattacharyaGGM10,CaiD11,CaiDW12, CaiDW12b, CaiH12, ChawlaHK07,ChawlaHMS10,DW12,HartN12,KleinbergW12} studies a special case of this problem (perhaps replacing Bayesian Incentive Compatibility with Incentive Compatibility). A detailed review of this literature is given in Section~\ref{sec:related}.

\smallskip\noindent  \framebox{
\begin{minipage}{0.46\textwidth}
\textbf{Revenue-Maximizing Multi-Dimensional Mechanism Design Problem (MDMDP):} Given $m$ distributions  $\mathcal{D}_1,\ldots,\mathcal{D}_m$, supported on $\mathbb{R}^n$, over valuation vectors for $n$ heterogenous items (possibly correlated across items), and feasibility constraints $\mathcal{F}$, output a Bayesian Incentive Compatible (BIC)\footnote{A mechanism is said to be BIC if it asks bidders to report their valuation to the mechanism and it is in each bidder's best interest to report truthfully, given that every other bidder does so as well. See Section~\ref{sec:notation} for a formal definition.} mechanism $M$ whose allocation is in $\mathcal{F}$ with probability $1$ and whose expected revenue is optimal relative to any other, possibly randomized, BIC mechanism when played by $m$ additive bidders\footnote{A bidder is additive if their value for a bundle of items is the sum of their value for each item in the bundle.} whose valuation vectors are sampled from $\mathcal{D} = \times_i \mathcal{D}_i$.
\end{minipage}}

\smallskip It was recently shown that solving MDMDP under feasibility constraints $\mathcal{F}$ can be poly-time reduced to (the algorithmic problem of) maximizing social welfare under the same feasibility constraints $\cal F$, i.e. running the VCG allocation rule with constraints ${\cal F}$~\cite{CaiDW12b}. This result implies that, for all ${\cal F}$'s such that maximizing social welfare can be solved efficiently, MDMDP can also be solved efficiently. On the other hand, the reduction of~\cite{CaiDW12b} is geometric and sensitive to having an exact algorithm for maximizing welfare, and this limits the span of mechanism design settings that can be tackled. In this work we extend this reduction, making it robust to approximation. Namely, we reduce the problem of approximating MDMDP to within a factor $\alpha$ to the problem of approximately optimizing social welfare  to within the same factor $\alpha$. Before  stating our result formally, let us define the concept of a {\em virtual implementation} of an algorithm.

\begin{definition}\label{def: virtual implementation}
Let $A$ be a social welfare algorithm, i.e. an algorithm that takes as input a vector $(t_1,\ldots,t_m)$ of valuations (or types) of bidders and outputs an allocation $O \in {\cal F}$. A {\em virtual implementation of $A$} is defined by a collection of functions $f_1,\ldots,f_m$, such that $f_i : T_i \rightarrow \mathbb{R}^n$, where $T_i$ is bidder $i$'s type set. On input $(t_1,\ldots,t_m)$ the virtual implementation outputs $A(f_1(t_1),\ldots,f_m(t_m))$, i.e. instead of running $A$ on the ``real input'' $(t_1,\ldots,t_m)$ it runs the algorithm on the ``virtual input'' $(f_1(t_1),\ldots,f_m(t_m))$ defined by the functions $f_1,\ldots,f_m$. The functions $f_1,\ldots,f_m$ are called \emph{virtual transformations}.
\end{definition}


\noindent With this definition, we state our main result informally below, and formally as Theorem~\ref{thm:main} of Section~\ref{sec:theorems}. 

\begin{informaltheorem}\label{infthm:appx} Fix some arbitrary ${\cal F}$ and finite $T_1,\ldots,T_m$ and let $A: \times_i T_i \rightarrow {\cal F}$ be a (possibly randomized, not necessarily truthful) social welfare algorithm, whose output is in $\mathcal{F}$ with probability $1$. Suppose that, for some $\alpha \leq 1$, $A$ is an $\alpha$-approximation algorithm to the social welfare optimization problem for ${\cal F}$, i.e. on all inputs $\vec{t}$ the allocation output by $A$ has social welfare that is within a factor of $\alpha$ from the optimum for $\vec{t}$. Then {for all $\mathcal{D}_1,\ldots,\mathcal{D}_m$} supported on $T_1,\ldots,T_m$ respectively, and all $\epsilon >0$, given black-box access to $A$ and without knowledge of ${\cal F}$, we can obtain  an $(\alpha-\epsilon)$-approximation algorithm for  MDMDP whose runtime is polynomial in the number of items, the number of bidder types (and \emph{not} type profiles), and the runtime of $A$. Moreover, the allocation rule of the output mechanism is a distribution over virtual implementations of $A$. 
\end{informaltheorem}
\noindent In addition to our main theorem, we provide in Section~\ref{sec:theorems} extensions for distributions of infinite support and improved runtimes in certain cases, making use of techniques from~\cite{DW12}. We also show that our results still hold even in the presence of bidders with hard budget constraints. {We remark that the functions defining a virtual implementation of a social welfare algorithm (Definition~\ref{def: virtual implementation}) may map a bidder type to a vector with negative coordinates. We require that the approximation guarantee of the given social welfare algorithm is still valid for inputs with negative coordinates. This is not a restriction for arbitrary downwards-closed ${\cal F}$'s, as any $\alpha$-factor approximation algorithm that works for non-negative vectors can easily be (in a black-box way) converted to an $\alpha$-factor approximation algorithm allowing arbitrary inputs.\footnote{The following simple black-box transformation achieves this: first zero-out all negative coordinates in the input vectors; then call the approximation algorithm; in the allocation output by the algorithm un-allocate item $j$ from bidder $i$ if the corresponding coordinate is negative; this is still a feasible allocation as the setting is downwards-closed.} But this is not necessarily true for non downwards-closed ${\cal F}$'s. If optimal social welfare cannot be tractably approximated (without concern for truthfulness) under arbitrary inputs, our result is not applicable.}

\paragraph{Beyond Additive Settings:} We note that the additivity assumption on the bidders' values for bundles of items is already general enough to model all settings that have been studied in the revenue-maximizing literature cited above, and already contains all unit-demand settings. 

{Beyond these settings that are already additive, we remark that we can easily extend our results to broader settings with  minimal loss in computational efficiency. As an easy example, consider a single-minded combinatorial auction where bidder $i$ is only interested  in receiving some fixed subset  $S_i$ of items, or nothing, and has (private) value $v_i$ for $S_i$. Instead of designing an auction for the original setting, we can design an auction for a single ``meta-item'' such that allocating the meta-item to bidder $i$ means allocating  subset $S_i$ to bidder $i$. So bidder $i$ has value $v_i$ for the meta-item. The meta-item can be simultaneously allocated to several bidders. However, to faithfully represent the underlying setting, we define our feasibility constraints to enforce that we never simultaneously allocate the meta-item to bidders $i$ and $j$ if $S_i \cap S_j \neq \emptyset$. As there is now only one item, the bidders are trivially additive. So, the new setting faithfully represents the original setting, there is only $1$ item, and the bidders are additive. So we can use our main theorem to solve this setting efficiently.}


More generally, we can define the notion of \emph{additive dimension} of an auction setting to be the minimum number of meta-items required so that the above kind of transformation can be applied to yield an equivalent setting whose bidders are additive. For example, the additive dimension of any setting with arbitrary feasibility constraints and additive bidders with arbitrary demand constraints is $n$.  The additive dimension of a single-minded combinatorial auction setting is $1$. The additive dimension of general (i.e. non single-minded) combinatorial auction settings, as well as all settings with risk-neutral bidders is at most $2^n$ (make a meta-item for each possible subset of items). In Appendix~\ref{sec:dimension} we discuss the following observation and give examples of settings with low additive dimension, including settings where bidders have symmetric submodular valuations~\cite{BadanidiyuruKS12}.

\begin{observation}\label{obs:dimension} In any setting with additive dimension $d$, Informal Theorem~\ref{infthm:appx} holds after multiplying the runtime by a $\poly (d)$ factor, assuming that the transformation to the additive representation of the setting can be carried out computationally efficiently in the setting's specification.
\end{observation}

\subsection{Approach and Techniques.} \label{sec:techniques}

Our main result, as well as those of~\cite{AlaeiFHHM12,CaiDW12,CaiDW12b}, are enabled by an algorithmic characterization of {\em interim allocation rules} of auctions.\footnote{The interim rule of an auction is the collection of marginal allocation probabilities $\pi_{ij}(t_i)$, defined for each item $j$, bidder $i$, and type $t_i$ of that bidder, representing the probability that item $j$ is allocated to bidder $i$ when her type is $t_i$, and in expectation over the other bidders' types, the randomness in the mechanism, and the bidders' equilibrium behavior. See Section~\ref{sec:notation}.}  The benefit of working with the interim rule is, of course, the exponential (in the number of bidders) gain in description complexity that it provides compared to the ex post allocation rule, which specifies the behavior of the mechanism for every {\em vector} of bidders' types. On the other hand, checking whether a given interim rule is consistent with an auction is a non-trivial task. Indeed, even in single-item settings, where a necessary and sufficient condition for feasibility of interim rules had been known for a while~\cite{Border91,Border07,CheKM11}, it was only recently that efficient algorithms were obtained~\cite{CaiDW12,AlaeiFHHM12}. These approaches also generalized to serving many copies of an item with a matroid feasibility constraint on which bidders can be served an item simultaneously~\cite{AlaeiFHHM12}, but for more general feasibility constraints there seemed to be an obstacle in even  defining necessary and sufficient conditions for feasibility~\cite{CaiDW12}, let alone checking them efficiently.

In view of this difficulty, it is quite surprising that a general approach for the problem was offered in~\cite{CaiDW12b}. The main realization was that, for arbitrary feasibility constraints, the set of feasible interim  rules is a convex polytope, whose facets are accessible via black-box calls to an exact welfare optimizer for the same feasibility constraints. Such an algorithm can be turned into a separation oracle for the polytope and used to optimize over it with Ellipsoid. However, this approach requires use of an exact optimizer for welfare, making it computationally intractable in settings where optimal social welfare can only be tractably approximated. 

Given only an approximation algorithm for optimizing social welfare, one cannot pin down the facets of the polytope of feasible interim rules exactly. Still, a natural approach could be to resign from the exact polytope of feasible interim rules, and let the approximation algorithm define a large enough sub-polytope. Namely, whenever the separation oracle of~\cite{CaiDW12b} uses the output of the social welfare optimizer to define a facet, make instead a call to the social welfare approximator and use its output to define the facet. Unfortunately, unless the approximation algorithm is a maximal-in-range algorithm, the separation oracle obtained does not necessarily define a polytope. In fact, the region is likely not even convex, taking away all the  geometry that is crucial for applying Ellipsoid.

Despite this, we show that ignoring the potential non-convexity, and running Ellipsoid with this ``weird separation oracle'' (called ``weird'' because it does not define a convex region) gives an approximation guarantee anyway, allowing us to find an approximately optimal interim rule with black-box access to the social welfare approximator. The next difficulty is that, after we find the approximately optimal interim rule, we still need to find an auction implementing it. In~\cite{CaiDW12b} this is done via a geometric algorithm that decomposes a point in the polytope of feasible interim rules into a convex combination of its corners. Now that we have no polytope to work with, we have no hope of completing this task. Instead, we show that for any point $\vec{\pi}$ deemed feasible by our weird separation oracle, the black-box calls made during the execution to the social welfare approximator contain enough information to decompose $\vec{\pi}$ into a convex combination of virtual implementations of the approximation algorithm (which are not necessarily extreme points, or even contained in the region defined by our weird separation oracle). After replacing the separation oracle of~\cite{CaiDW12b} with our weird separation oracle, and the decomposition algorithm with this new decomposition approach, we obtain the proof of our main theorem (Informal Theorem~\ref{infthm:appx} above, and Theorem~\ref{thm:main} in Section~\ref{sec:theorems}). Our approach is detailed in Sections~\ref{sec:WSO},~\ref{sec:revenue} and~\ref{sec:runtime}. 

\subsection{Related Work} \label{sec:related}
\subsubsection{Optimal Mechanism Design.}

In his seminal paper, Myerson solved the single-item case of the MDMDP~\cite{Myerson81}. Shortly after, the result was extended to all ``single-dimensional settings,'' where the seller has multiple copies of the same item and some feasibility constraint ${\cal F}$ on which of the bidders can simultaneously receive a copy. The algorithmic consequence of these results is that, for all ${\cal F}$'s such that social welfare can be (not necessarily truthfully) efficiently optimized, the revenue-optimal auction can also be efficiently found, and run. On the other hand, before this work, there was no analogue of this for approximation algorithms, allowing a generic reduction from revenue approximation to (not necessarily truthful) social-welfare approximation.

On the multi-dimensional front, where there are multiple, heterogeneous items for sale, progress had been slower~\cite{ManelliV07}, and only recently computationally efficient constant factor approximations for special cases were obtained~\cite{ChawlaHK07,BhattacharyaGGM10,ChawlaHMS10,Alaei11,KleinbergW12,RoughgardenTY12}. These results cover settings where the bidders are unit-demand and the seller has matroid or matroid-intersection constraints on which bidders can simultaneously receive items, or the case of additive-capacitated bidders, i.e. settings that are special cases of the MDMDP framework.\footnote{In some of these results, bidders may also have  budget constraints (and this does not directly fit in the MDMDP framework). Nevertheless, budgets can be easily incorporated to the framework without any loss, as was shown in~\cite{CaiDW12b}.} More recently, computationally efficient optimal solutions were obtained for even more restricted cases of MDMDP~\cite{AlaeiFHHM12,CaiDW12,DW12}, until a general, computationally efficient reduction  from revenue to welfare optimization was given in~\cite{CaiDW12b}. This result offers the analog of Myerson's result for multi-dimensional settings. Nevertheless, the question still remained whether there is an approximation preserving reduction from revenue  to (not necessarily truthful) welfare optimization. This reduction is precisely what this work provides, resulting in approximately optimal solutions to MDMDP for all settings where maximizing welfare is intractable, but approximately optimizing welfare (without concern for truthfulness) is tractable.


\subsubsection{Black-Box Reductions in Mechanism Design.}

Our reduction from approximate revenue optimization to non-truthful welfare approximation is a black-box reduction. Such reductions have been a recurring theme in mechanism design literature but only for {\em welfare}, where approximation-preserving reductions from truthful welfare maximization to non-truthful welfare maximization have been provided~\cite{BriestKV05,BabaioffLP06,HartlineL10,DughmiR10,BeiH11,HartlineKM11}. The techniques used here are orthogonal to the main techniques of these works. In the realm of black-box reductions in mechanism design, our work is best viewed as ``catching up'' the field of revenue maximization to welfare maximization, for the settings covered by the MDMDP framework.

\subsubsection{Weird Separation Oracle, Approximation, and Revenue Optimization.}

Gr{\"o}tschel et al.~\cite{GLS} show that exactly optimizing any linear function over a bounded polytope $P$ is equivalent to having a separation oracle for $P$. This is known as the equivalence of exact separation and optimization. Jansen extends this result to accommodate approximation \cite{Klaus}. He shows that given an approximation algorithm $\mathcal{A}$ such that for any direction $\vec{w}$, $\mathcal{A}$ returns an approximately extreme point $\mathcal{A}(\vec{w})\in P$, where $\mathcal{A}(\vec{w})\cdot \vec{w} \geq \alpha \cdot \max_{\vec{v} \in P}\{ \vec{w} \cdot \vec{v} \}$, one can construct a strong, approximate separation oracle which either asserts that a given point $\vec{x}\in P$ or outputs a hyperplane that separates $\vec{x}$ from $\alpha P$ (the polytope $P$ shrunk by $\alpha$). We show a similar but stronger result. Under the same conditions, our weird separation oracle either outputs a hyperplane separating $\vec{x}$ from a polytope $P_1$ that contains $\alpha P$, or asserts that $\vec{x} \in P_2$, where $P_2$ is a polytope contained in $P$. A precise definition of $P_1$ and $P_2$ is given in Section~\ref{subsec:threepolytopes}. Moreover, for any point $\vec{x}$ that the weird separation oracle asserts is in $P_2$, we show how to decompose it into a convex combination of points of the form $\mathcal{A}(\vec{w})$. This is crucial for us, as our goal is not just to find an approximately optimal reduced form, but also to implement it. The technology of~\cite{Klaus} is not enough to accomplish this, which motivates our stronger results. 

But there is another, crucial reason that prevents using the results of~\cite{Klaus}, and for that matter~\cite{GLS} (for the case $\alpha=1$), as a black box for our purposes. Given a computationally efficient, $\alpha$-approximate social-welfare algorithm $A$ for feasibility constraints ${\cal F}$, we are interested in obtaining a separation oracle for the polytope $P=F({\cal F}, {\cal D})$ of feasible interim allocation rules of auctions that respect ${\cal F}$ when bidder types are drawn from distribution ${\cal D}$. To use~\cite{Klaus} we need to use $A$ to come up with an $\alpha$-approximate linear optimization algorithm for $P$. But, in fact, we do not know how to find such an algorithm efficiently for general ${\cal F}$, due to the exponentiality of the support of ${\cal D}$ (which is a product distribution over ${\cal D}_1,\ldots,{\cal D}_m$). Indeed, given $\vec{w}$ we only know how to query $A$ to obtain some $\pi^*(\vec{w})$ such that $\pi^*(\vec{w}) \cdot \vec{w} \geq \alpha \cdot \max_{\vec{\pi} \in P}\{ \vec{w} \cdot \vec{\pi} \} - \epsilon$, for some small $\epsilon >0$. This additive approximation error that enters the approximation guarantee of our linear optimization algorithm is not compatible with using the results of~\cite{Klaus} or~\cite{GLS} as a black box, and requires us to provide our own separation to optimization reduction, together with additional optimization tools.

\section{Preliminaries and notation}\label{sec:notation}
\subsection{MDMDP.}

Here are preliminaries and notation regarding the MDMDP aspect of our results. We use the same notation as~\cite{CaiDW12b}. Denote the number of bidders by $m$, the number of items by $n$, and the type space of bidder $i$ by $T_i$. To ease notation, we sometimes use $B$ ($C$, $D$, etc.) to denote possible types of a bidder (i.e. elements of $T_i$), and use $t_i$ for the random variable representing the instantiated type of bidder $i$. So when we write $\Pr[t_i = B]$, we mean the probability that bidder $i$'s type is $B$. The elements of $\times_i T_i$ are called {\em type profiles}, and specify a type for every bidder. We assume type profiles are sampled from a known distribution ${\cal D}$ over $\times_i T_i$. We denote by ${\cal D}_i$ the marginal of this distribution on bidder $i$'s type, and use {${\cal D}_{-i}(B)$} to denote the marginal of ${\cal D}$ over the types of all bidders except $i$, conditioned on $t_i = B$. If $\mathcal{D}$ is a product distribution, we will drop the parameter $B$ and just write $\mathcal{D}_{-i}$. We refer the reader to Appendix~\ref{sec:input distribution} for a discussion on how an algorithm might access the distribution $\mathcal{D}$. 

Let $ [m]\times[n]$ denote the set of possible \emph{assignments} (i.e. the element $(i,j)$ denotes that bidder $i$ was awarded item $j$). {We call {(distributions over)} subsets of $[m]\times [n]$ {(randomized)} \emph{allocations}, and functions mapping {type} profiles to {(possibly randomized)} allocations \emph{allocation rules}. We call an allocation combined with a price charged to each bidder an \emph{outcome}, and an allocation rule combined with a pricing rule a {(direct revelation, or direct)} \emph{mechanism}.} {As discussed in Section~\ref{sec:intro}, we may also have a set system $\mathcal{F}$ on $[m]\times[n]$ (that is, a subset of $2^{[m] \times [n]}$), encoding constraints on what allocations are feasible. ${\cal F}$ may be incorporating arbitrary demand constraints imposed by each bidder, and supply constraints imposed by the seller, and will be referred to as our {\em feasibility constraints}. In this case, we restrict all allocation rules to be supported on ${\cal F}$. We always assume that $\emptyset \in {\cal F}$, i.e. the auctioneer has the option to allocate no item.}

The interim allocation rule, also called \emph{reduced form} of an allocation rule,  is a vector function $\pi(\cdot)$, specifying values $\pi_{ij}(B)$, for all items $j$, bidders $i$ and types $B \in T_i$. $\pi_{ij}(B)$ is the probability that bidder $i$ receives item $j$ when truthfully reporting type $B$, where the probability is over the randomness of all other bidders' types {(drawn from ${\cal D}_{-i}(B)$)} and the internal randomness of the allocation rule, assuming that the other bidders report truthfully their types. Sometimes, we will want to think of the reduced form as a $n\sum_{i=1}^{m} |T_i|$-dimensional vector, and may write $\vec{\pi}$ to emphasize this view. To ease notation we will also denote by $T := n\sum_i |T_i|$.

Given a reduced form $\pi$, we will be interested in whether the form is ``feasible'', or can be ``implemented.''  By this we mean designing a feasible allocation rule $M$ (i.e. one that respects feasibility constraints ${\cal F}$ on every type profile with probability $1$ over the randomness of the allocation rule) such that the probability that bidder $i$ receives item $j$ when truthfully reporting type $B$ is exactly $\pi_{ij}(B)$, where the probability is computed with respect to the randomness in the allocation rule and the randomness in the types of the other bidders (drawn from ${\cal D}_{-i}(B)$), assuming that the other bidders report truthfully.  While viewing reduced forms as vectors, we denote by $F(\mathcal{F},\mathcal{D})$ the set of feasible reduced forms when the feasibility constraints are $\mathcal{F}$ and bidder types are sampled from $\mathcal{D}$. 

A bidder is {\em additive} if her value for a bundle of items is the sum of her values for the items in that bundle. To specify the preferences of additive bidder $i$, we can provide a valuation vector $\vec{v}_i$, with the convention that $v_{ij}$ represents her value for item $j$.  Even in the presence of arbitrary (possibly combinatorial) demand constraints, the \emph{value} of bidder $i$ of type $\vec{v}_i$ for a randomized allocation that respects the bidder's demand constraints with probability $1$, and whose expected probability of allocating item $j$ to the bidder is $\pi_{ij}$, is just the bidder's expected value, namely $\sum_j v_{ij} \cdot \pi_{ij}$. The \emph{utility} of bidder $i$ for the same allocation when paying price $p_i$ is just $\sum_j v_{ij} \cdot \pi_{ij} - p_i$. Such bidders whose value for a distribution of allocations is their expected value for the sampled allocation are called \emph{risk-neutral}. Bidders subtracting price from expected value are called \emph{quasi-linear}. 

Throughout the paper we denote by $A$ a (possibly randomized, non-truthful) social welfare algorithm that achieves an $\alpha$-fraction of the optimal welfare for feasibility constraints $\mathcal{F}$. We denote by $A(\{f_i\}_i)$ the virtual implementation of $A$ with virtual transformations $f_i$ (see Definition~\ref{def: virtual implementation}).

Some arguments will involve reasoning about the \emph{bit complexity} of a rational number. We say that a rational number has bit complexity $b$ if it can be written with a binary numerator and denominator that each have at most $b$ bits. {We also take the bit complexity of a rational vector to be the total number of bits required to describe its coordinates. Similarly, the bit complexity of an explicit distribution supported on rational numbers with rational probabilities is the total number of bits required to describe the points in the support of the distribution and the probabilities assigned to each point in the support. For our purposes the bidder distributions ${\cal D}_1, \ldots, {\cal D}_m$ are given explicitly, while ${\cal D}=\times_i {\cal D}_i$ is described implicitly as the product of ${\cal D}_1,\ldots,{\cal D}_m$. }

Finally, for completeness, we define in Appendix~\ref{app:prelims} the standard notions of Bayesian Incentive Compatibility (BIC) and Individual Rationality (IR) of mechanisms.

\subsection{Weird Separation Oracles.}
In our technical sections, we will make use of ``running the ellipsoid algorithm with a weird separation oracle.'' A weird separation oracle is just an algorithm that, on input $\vec{x}$, either outputs ``yes,'' or a hyperplane that $\vec{x}$ violates. We call it ``weird'' because the set of points that will it accepts is not necessarily convex, or even connected, so it is not a priori clear what it means to run the ellipsoid algorithm with a weird separation oracle. When we say ``run the ellipsoid algorithm with a weird separation oracle'' we mean:

\begin{enumerate}
\item Find a meaningful ellipsoid to start with (this will be obvious for all weird separation oracles we define, so we will not explicitly address this).
\item Query the weird separation oracle on the center of the current ellipsoid. If it is accepted, output it as a feasible point. Otherwise, update the ellipsoid using the violated hyperplane (in the same manner that the standard ellipsoid algorithm works).
\item Repeat step 2) for a pre-determined number of iterations $N$ ($N$ will be chosen appropriately for each weird separation oracle we define). If a feasible point is not found after $N$ iterations, output ``infeasible.''
\end{enumerate}

It is also important to note that we are \emph{not} using the ellipsoid algorithm as a means to learning whether some non-convex set is empty. We are using properties of the ellipsoid algorithm with carefully chosen weird separation oracles to learn information, not necessarily related to a feasibility question.
\section{The Weird Separation Oracle (WSO)}\label{sec:WSO}
	In this section, we take a detour from mechanism design, showing how to construct a weird separation oracle from an algorithm that approximately optimizes linear functions over a convex polytope. Specifically, let $P$ be a bounded polytope containing the origin, and let $\mathcal{A}$ be any algorithm that takes as input a linear function $f$ and outputs a point $\vec{x} \in P$ that approximately optimizes $f$ (over $P$). We will define our weird separation oracle using black-box access to $\mathcal{A}$ and prove several useful properties that will be used in future sections. We begin by discussing three interesting convex regions related to $P$ in Section~\ref{subsec:threepolytopes}. This discussion provides insight behind why we might expect $WSO$ to behave reasonably. In addition, the polytopes discussed will appear in later proofs. In Section~\ref{subsec:WSO} we define $WSO$ formally and prove several useful facts about executing the ellipsoid algorithm with $WSO$. For this section, we will not address running times, deferring this to Section~\ref{sec:runtime}. Our basic objects for this section are encapsulated in the following definition. 

\begin{definition}\label{def:avcg}
		$P$ is a convex $d$-dimensional polytope contained in $[-1,1]^d$, $\alpha\leq 1$ is an absolute constant, and $\cal{A}$ is an approximation algorithm such that  for any $\vec{w}\in \mathbb{R}^d$, $\mathcal{A}(\vec{w})\in P$ and $\mathcal{A}(\vec{w})\cdot \vec{w} \geq \alpha\cdot\max_{\vec{x}\in P}\{ \vec{x}\cdot\vec{w}\}$. (Since $\vec{0}\in P$, this is always non-negative.)\end{definition}
	
Notice that the restriction that $P \subseteq [-1,1]^d$ is without loss of generality as long as $P$ is bounded, as in this section we deal with multiplicative approximations. 


\subsection{Three Convex Regions.}\label{subsec:threepolytopes}
Consider the following convex regions, where $Conv(S)$ denotes the convex hull of $S$.
		
\begin{itemize}
		\item $P_0=\{\vec{\pi}\ |\ \frac{\vec{\pi}}{\alpha}\in P\}$.
		\item $P_1 = \{\vec{\pi}\ |\ \vec{\pi}\cdot\vec{w}\leq \mathcal{A}(\vec{w})\cdot\vec{w},\ \forall \vec{w}\in[-1,1]^d\}$.
		\item $P_2 = Conv(\{\mathcal{A}(\vec{w}),\ \forall \vec{w}\in[-1,1]^d\})$. 
\end{itemize}

It is not hard to see that, if $\mathcal{A}$ always outputs the exact optimum (i.e. $\alpha = 1$), then all three regions are the same. It is this fact that enables the equivalence of separation and optimization~\cite{GLS}. It is not obvious, but perhaps not difficult to see also that if $\mathcal{A}$ is a maximal-in-range algorithm,\footnote{Let $S$ denote the set of vectors that are ever output by $\mathcal{A}$ on any input. Then $\mathcal{A}$ is maximal-in-range if, for all $\vec{w}$, $\mathcal{A}(\vec{w}) \in \argmax_{\vec{x} \in S}\{\mathcal{A}(\vec{x}) \cdot \vec{w} \}$.} then $P_1 = P_2$. It turns out that in this case, $WSO$ (as defined in Section~\ref{subsec:WSO}) actually defines a polytope. We will not prove this as it is not relevant to our results, but it is worth observing where the complexity comes from. We conclude this section with a quick lemma about these regions, whose proof appears in Appendix~\ref{app:threepolytopes}.

\begin{lemma}\label{lem:chain}
	$P_0\subseteq P_1\subseteq P_2$.
	\end{lemma}
	
	\subsection{WSO.}\label{subsec:WSO}
	
	Before defining $WSO$, let's state the properties we want it to have. First, for any challenge $\vec{\pi}$, $WSO$ should either assert $\vec{\pi}\in P_{2}$ or output a hyperplane separating $\vec{\pi}$ from $P_{1}$. Second, for any $\vec{\pi}$ such that $WSO(\vec{\pi})=``yes''$, we should be able to decompose $\vec{\pi}$ into a convex combination of points of the form $\mathcal{A}(\vec{w})$. Why do we want these properties? Our goal in later sections is to write a LP that will use $WSO$ for $F(\mathcal{F},\mathcal{D})$ to find a reduced form auction whose revenue is at least $\alpha \opt$. Afterwards, we have to find an actual mechanism that implements this reduced form. So $WSO$ needs to guarantee two things: First, running a revenue maximizing LP with $WSO$ must terminate in a reduced form with good revenue. Second, we must be able to implement any reduced form that $WSO$ deems feasible. Both claims will be proved in Section~\ref{sec:revenue} using lemmas proved here. That using $WSO$ achieves good revenue begins with Fact~\ref{fact:WSOpolytope}. That we can implement any reduced form deemed feasible by $WSO$ begins with Lemma~\ref{lem:convhull}. We define $WSO$ in Figure~\ref{fig:WSO}. 
	
		\begin{figure}[h!]
	\colorbox{MyGray}{
	\begin{minipage}{0.46\textwidth} {	
		$WSO(\vec{\pi})=$
	\begin{itemize}	
	\item ``\textbf{Yes}'' if the ellipsoid algorithm with $N$ iterations\footnote{The appropriate choice of $N$ for our use of $WSO$ is provided in Corollary~\ref{cor:N} of Section~\ref{sec:runtime}. The only place that requires an appropriate choice of $N$ is the proof of Lemma~\ref{lem:convhull}.} outputs ``infeasible'' on the following problem:
			 
			  \underline{\textbf{variables:}} $\vec{w}, t$;
			 
			 \underline{\textbf{constraints:}}
				\begin{itemize}
		 		\item $\vec{w} \in [-1,1]^d;$
		 		\item $t - \vec{\pi} \cdot \vec{w} \leq -\delta$;\footnote{The appropriate choice of $\delta$ for our use of $WSO$ is provided in Lemma~\ref{lem:delta} of Section~\ref{sec:runtime}. The only place that requires an appropriate choice of $\delta$ is the proof of Lemma~\ref{lem:convhull}.}
		 		\item $\widehat{WSO}(\vec{w},t) = $
		 		\begin{itemize}
		 		\item ``yes'' if $t \geq \mathcal{A}(\vec{w}) \cdot \vec{w}$;\footnote{Notice that the set $\{(\vec{w},t)|\widehat{WSO}(\vec{w},t) =$ ``Yes''$\}$ is not necessarily convex or even connected.}
		 		\item the violated hyperplane $t' \geq \mathcal{A}(\vec{w})\cdot \vec{w}'$ otherwise.
		 		\end{itemize}
				 		\end{itemize}
	\item If a feasible point $(t^*,\vec{w}^*)$ is found, output the violated hyperplane $\vec{w}^* \cdot \vec{\pi}' \leq t^*$.
	\end{itemize}		}
			\end{minipage}} \caption{A ``weird'' separation oracle.}\label{fig:WSO}
	\end{figure}
	
	Let's now understand what exactly $WSO$ is trying to do. What $WSO$ really wants is to act as a separation oracle for $P_2$. As $P_2$ is a polytope, if $\vec{\pi} \notin P_2$, then there is some weight vector $\vec{w}$ such that $\vec{\pi} \cdot \vec{w} > \max_{\vec{x} \in P_2}\{\vec{x} \cdot \vec{w}\}$. $WSO$ wants to find such a weight vector or prove that none exists (and therefore $\vec{\pi} \in P_2$). It is shown in~\cite{GLS} that if we replace $\mathcal{A}(\vec{w})$ with $\argmax_{\vec{x} \in P_2} \{\vec{x} \cdot \vec{w} \}$ inside $\widehat{WSO}$, then $WSO$ would be a separation oracle for $P_2$.	Unfortunately, unless $\mathcal{A}$ is maximal-in-range, we cannot find $\argmax_{\vec{x} \in P_2} \{\vec{x} \cdot \vec{w} \}$ with only black-box access to $\mathcal{A}$.\footnote{If $\mathcal{A}$ is maximal-in-range, then this is exactly $\mathcal{A}(\vec{w})$.} So $WSO$ makes its best guess that $\mathcal{A}(\vec{w})$ is the maximizer it is looking for. Of course, this is not necessarily the case, and this is why the set of points accepted by $WSO$ is not necessarily a convex region. Now, we need to prove some facts about $WSO$ despite this. All proofs can be found in Appendix~\ref{subapp:WSO}.
	
	\begin{fact}\label{fact:WSOpolytope} 	Consider an execution of the ellipsoid algorithm using $WSO$, possibly together with additional variables and constraints. Let $Q$ be the polytope defined by the halfspaces output by $WSO$ during its execution. Then during the entire execution, $P_1 \subseteq Q$.
	\end{fact}

\begin{fact}\label{fact:P1} If $\vec{\pi} \in P_1$, then $WSO(\vec{\pi}) = $ ``yes.''
\end{fact}

\begin{corollary}\label{cor:WSOandP1}
When $WSO$ rejects $\vec{\pi}$, it acts as a valid separation oracle for $P_1$, or any polytope contained in $P_1$ (i.e. the hyerplane output truly separates $\vec{\pi}$ from $P_1$). In other words, the only difference between $WSO$ and a valid separation oracle for $P_1$ is that $WSO$ may accept points outside of $P_1$.
\end{corollary}

	\begin{lemma}\label{lem:convhull} Let $WSO(\vec{\pi}) = $ ``yes'' and let $S$ denote the set of weights $\vec{w}$ such that $WSO$ queried $\widehat{WSO}(\vec{w},t)$ for some $t$ during its execution. Then $\vec{\pi} \in Conv(\{\mathcal{A}(\vec{w}) | \vec{w} \in S\})$. 
	\end{lemma}

	
\section{Approximately Maximizing Revenue using $WSO$}\label{sec:revenue}
In this section, we show that running the revenue maximizing LP of~\cite{CaiDW12b} using the weird separation oracle of the previous section obtains good revenue, and outputs a reduced form that can be implemented with only black-box access to a social welfare algorithm $A$. 

In brush strokes, the approach of~\cite{CaiDW12b} is the following. They start by creating a proxy distribution $\mathcal{D}'$ that is a (correlated across bidders) uniform distribution over $\poly(n,T,1/\epsilon)$ type profiles. {Roughly,} $\mathcal{D}'$ is obtained by sampling the same number of profiles from $\mathcal{D}$, and forming the uniform distribution over them, and its advantage over ${\cal D}$ is that its support is polynomial. With ${\cal D}'$ at hand, it shown that the LP of Figure~\ref{fig:revenue benchmark} in Appendix~\ref{app:revenue} outputs a reduced form whose revenue is at least $\opt - \epsilon$. This is proved by showing that the polytopes $F(\mathcal{F},\mathcal{D})$ and $F(\mathcal{F},\mathcal{D}')$ are ``$\epsilon$-close'' in some meaningful way. To show how we adapt this approach to our setting, we need a definition.

\begin{definition}\label{def:avcg2}
			 Let $\vec{w} \in \mathbb{R}^{T}$, and $\hat{\mathcal{D}}$ be a {(possibly correlated)} distribution over bidder type profiles. Define $f_i: T_i \rightarrow \mathbb{R}^n$ so that $f_{ij}(B) = \frac{w_{ij}(B)}{\Pr[t_i = B]}$. Then $A_{\hat{\cal{D}}}(\vec{w})$ denotes the allocation rule $A(\{f_i\}_i)$, $R^A_{\hat{\cal{D}}}(\vec{w})$ denotes the reduced form of $A_{\hat{\cal{D}}}(\vec{w})$, and $W^A_{\hat{\cal{D}}}(\vec{w}) := R^A_{\hat{\cal{D}}}(\vec{w}) \cdot \vec{w}$ is exactly the expected virtual welfare obtained by algorithm $A$ under the virtual transformations $\{f_i\}_i$. For the purpose of the dot product, recall that we may view reduced forms as $T$-dimensional vectors (Section~\ref{sec:notation}).
		\end{definition}
		
Given this definition, {and for the same ${\cal D}'$ used in~\cite{CaiDW12b},} we let $P=F(\mathcal{F},\mathcal{D}')$, and $\mathcal{A}(\vec{w})$ be the algorithm that on input $\vec{w}\in \mathbb{R}^{T}$ returns $R^A_{\cal{D}'}(\vec{w})$. Because taking a dot product with $\vec{w}$ is exactly computing expected virtual welfare (as in Definition~\ref{def:avcg2}), it is clear that $\mathcal{A}$ is an $\alpha$-factor approximation algorithm for optimizing any linear function $\vec{w} \cdot \vec{x}$ over $\vec{x} \in P$. Using ${\cal A}$, we define $P_0$, $P_{1}$ and $P_{2}$ as in Section~\ref{sec:WSO}.

We continue to bound the revenue of the reduced form output by our LP of Figure~\ref{fig:revenue benchmark}. Denote by $Rev(F)$ the revenue obtained by the LP of Figure~\ref{fig:revenue benchmark}, and by $Rev(P_i)$ the revenue obtained by replacing $P$ with $P_i$. The proof of Lemma~\ref{lem:goodRevenue}, as well as all other results of this section are in Appendix~\ref{app:revenue}.

\begin{lemma}\label{lem:goodRevenue} $Rev(P_0) \geq \alpha Rev(F) \geq \alpha (\opt - \epsilon)$.
\end{lemma}

Now, denote by $Rev(WSO)$ the revenue obtained by replacing $P$ with $WSO$ in Figure~\ref{fig:revenue benchmark}. By ``replace $P$ with $WSO$,'' we mean run the optimization version of the ellipsoid algorithm that does a binary search on possible values for the objective function. On each subproblem {(i.e. for a guess $x$ of the revenue)}, run the ellipsoid algorithm using a new weird separation oracle $WSO'$, which does the following. For challenge $(\vec{\pi},\vec{p})$, first check if it satisfies the IR and BIC constraints in Figure~\ref{fig:revenue benchmark} and the revenue constraint $\sum_{i=1}^{m} \sum_{\vec{v}_i \in T_i} \Pr[t_i = \vec{v}_i]\cdot p_i(\vec{v}_i) \geq x$, for the guessed value $x$ of revenue. If not, output the hyperplane it violates. If yes, output $WSO(\vec{\pi})$. The ellipsoid algorithm will use \textbf{exactly the same parameters as if $WSO$ was a separation oracle for $P_0$.} In particular, we can calculate the number of iterations and the precision that Ellipsoid would use if it truly had access to a separation oracle for $P_0$,\footnote{These parameters were computed in~\cite{CaiDW12b} except for $F(\mathcal{F},\mathcal{D}')$ rather than $P_0$. As the latter is just the former scaled by $\alpha$ it is easy to modify these parameters to accommodate $\alpha$.  This is addressed in Lemma~\ref{lem:LPruntime} in Section~\ref{sec:runtime}.} and use the same number of iterations here. Moreover, we use here the same criterion for deeming the feasible region lower-dimensional that the Ellipsoid with separation oracle for $P_0$ would use. Similarly, the bit complexity over  values of $x$ that the binary search will search over is taken to be the same as if binary search and the ellipsoid algorithm were used to solve the LP of Figure~\ref{fig:revenue benchmark} with $P_0$ in place of $F(\mathcal{F},\mathcal{D}')$.

We now want to use Lemma~\ref{lem:goodRevenue} to lower bound $Rev(WSO)$. This is almost a direct corollary of Fact~\ref{fact:WSOpolytope}. The only remaining step is understanding the ellipsoid algorithm. 

\begin{proposition}\label{prop:WSOrevenue} If $x \leq Rev(P_0)$, then the ellipsoid algorithm using $WSO'$ (with the same parameters as if $WSO$ was a separation oracle for $P_{0}$) will always find a feasible point.
\end{proposition}

\begin{corollary}\label{cor:WSOrevenue} $Rev(WSO) \geq Rev(P_{0})$.
\end{corollary}

\begin{corollary}\label{cor:goodrevenue} $Rev(WSO) \geq \alpha (\opt - \epsilon)$.
\end{corollary}

Finally, we need to argue that we can implement any reduced form output by the LP with $WSO$, as otherwise the reduced form is useless. This is a direct consequence of Lemma~\ref{lem:convhull}:

\begin{corollary}\label{cor:implement} Let $\vec{\pi}^*$ denote the reduced form output by the LP of Figure~\ref{fig:revenue benchmark} using $WSO$ instead of $F(\mathcal{F},\mathcal{D}')$, and let $S$ be the set of weights $\vec{w}$ that are queried to $\widehat{WSO}$ during the execution. Then $\vec{\pi}^*$ can be implemented (for bidders sampled from $\mathcal{D}'$) as a distribution over virtual implementations of $A$ using only virtual transformations corresponding to weights in $S$.
\end{corollary}

At this point, we have shown that the reduced form $\vec{\pi}^*$ and pricing rule $p^*$ computed by the LP of Figure~\ref{fig:revenue benchmark} after replacing $F(\mathcal{F},\mathcal{D}')$ with $WSO$ achieves good revenue when bidders are sampled from $\mathcal{D}$, and define a BIC mechanism when bidders are sampled from $\mathcal{D}'$. We have also shown that we can implement $\vec{\pi}^*$ as a distribution over virtual implementations of $A$ using only weights that were queried during the execution of the LP, \emph{albeit for bidders are sampled from $\mathcal{D}'$}. 

The remaining step for correctness (we still have not addressed running time) is to show that, with high probability, the same distribution over virtual implementations of $A$ implements some reduced form $\vec{\pi}'$ when the bidders are sampled from $\mathcal{D}$ that satisfies $|\vec{\pi}^* - \vec{\pi}'|_1 \leq \epsilon$. Once we show this, we will have proved that our distribution over virtual implementations of $A$ and our pricing rule $p^*$ define an $\epsilon$-BIC, {$\epsilon$-IR} mechanism when bidders are sampled from $\mathcal{D}$ with good revenue. We will argue this informally in Appendix~\ref{app:epsBIC} and refer the reader to~\cite{CaiDW12b} for a formal proof of the same fact when using $F(\mathcal{F},\mathcal{D}')$ rather than $WSO$ in the LP of Figure~\ref{fig:revenue benchmark} as the proof is nearly identical. In addition, we can give every bidder type an $\epsilon$ rebate in order to get an $\epsilon$-BIC, IR mechanism for bidders sampled from $\mathcal{D}$ for an additional hit of $m\epsilon$ in revenue. (Recall that the runtime we are shooting for is polynomial in $1/\epsilon$, so $\epsilon$ can be made small enough to cancel the additional factor of $m$.) With this discussion, we have shown that our algorithm is correct: we have implemented some $\epsilon$-BIC, IR mechanism $(\vec{\pi}',\vec{p}^*-\epsilon)$ whose revenue is at least $\alpha(\opt - \epsilon)$. We show that our approach runs in polynomial time in Section~\ref{sec:runtime}.

\section{Runtime}\label{sec:runtime}
Until now, we have only established that our algorithms are correct, up to maybe choosing the right parameters in WSO, which was deferred to this section. Here, we set these parameters appropriately and analyze the running times of all our algorithms. In particular, we show that all reduced forms required in Section~\ref{sec:revenue} can be computed in polynomial time, and that both $WSO$ from Section~\ref{sec:WSO} and our revenue maximizing LP from Section~\ref{sec:revenue} run in polynomial time.\\ 

\paragraph{Analyzing WSO from Section~\ref{sec:WSO}.} We start with the appropriate choice of $\delta$. The proof of the following lemma is in Appendix~\ref{app:WSOruntime}. 
\begin{lemma}\label{lem:delta} Let $S$ be any subset of weight vectors in $[-1,1]^d$, $b$ be the bit complexity of $\vec{\pi}$, and $\ell$ {be an upper bound on the bit complexity of $\mathcal{A}(\vec{w})$ for all $\vec{w}\in[-1,1]^{d}$}. Then if $\vec{\pi} \notin Conv(\{\mathcal{A}(\vec{w}) |\vec{w} \in S\})$, there exists a weight vector $\vec{w}^*$ such that $\vec{\pi} \cdot \vec{w}^* \geq \max_{\vec{w} \in  S} \{\mathcal{A}(\vec{w}) \cdot \vec{w}^*\} + {4}\delta$, where $\delta = 2^{-\poly(d,\ell,b)}$ (does not depend on $S$).
\end{lemma}
\noindent The requirement that $\delta$ is chosen appropriately only appears in the proof of Lemma~\ref{lem:convhull}. As Lemma~\ref{lem:delta} describes an appropriate choice of $\delta$ for the proof to be correct, we take $\delta = 2^{-\poly(d,\ell,b)}$ in the definition of $WSO$. 

Next we address the appropriate choice of $N$ for the number of iterations used in $WSO$. This is stated in Corollary~\ref{cor:N}, and proved in Appendix~\ref{app:WSOruntime}.

\begin{corollary}\label{cor:N}
There exists some $N = \poly(d,\ell,b)$ such that, if $WSO$ has not found a feasible point after $N$ iterations of the ellipsoid algorithm, the following polytope ($P(S)$) is empty:
	$$t - \vec{\pi} \cdot \vec{w} \leq - \delta;$$
	$$t \geq \mathcal{A}(\vec{w}') \cdot \vec{w},\ \forall \vec{w}' \in S;$$
	$$\vec{w} \in [-1,1]^d;$$
where $S$ is the set of weights $\vec{w}'$ such that $WSO$ queried $\widehat{WSO}$ on $(t,\vec{w}')$ for some $t$ during its execution, $b$ is the bit complexity of $\vec{\pi}$, $\ell$ is {an upper bound on the bit complexity of $\mathcal{A}(\vec{w})$ for all $\vec{w}\in[-1,1]^{d}$}, and $\delta$ is chosen as in Lemma~\ref{lem:delta}.
\end{corollary}
Note that Lemma~\ref{lem:delta} and Corollary~\ref{cor:N} complete the description of $WSO$, and establish the truth of Lemma~\ref{lem:convhull}.

It remains to bound the running time of $WSO$. Let $rt_{\mathcal{A}}(x)$ be the running time of algorithm $\mathcal{A}$ on input whose bit complexity is $x$. With Lemma~\ref{lem:delta} and Corollary~\ref{cor:N}, we can bound the running time of $WSO$. This is stated below as Corollary~\ref{cor:WSOruntime} and proved in Appendix~\ref{app:WSOruntime}.

\begin{corollary}\label{cor:WSOruntime} Let $b$ denote the bit complexity of $\vec{\pi}$ and $\ell$ be an upper bound of the bit complexity of $\mathcal{A}(\vec{w})$ for all $\vec{w}\in[-1,1]^{d}$. Then on input $\vec{\pi}$, $WSO$ terminates in time $\poly(d,\ell,b,rt_{\mathcal{A}}(\poly(d,\ell,b)))$. 
\end{corollary}

\paragraph{Computing Reduced Forms.} In Section~\ref{sec:revenue} we need to use a possibly randomized social-welfare algorithm $A$ (to which we have black-box access) to obtain an $\alpha$-approximation algorithm $\mathcal{A}$ for optimizing any linear function $\vec{w} \cdot \vec{x}$ over $\vec{x} \in P=F(\mathcal{F},\mathcal{D}')$, where ${\cal D}'$ is a (correlated across bidders) uniform distribution over $\poly(n,T,1/\epsilon)$ type profiles. We need to argue that for a given input $\vec{w}\in \mathbb{R}^{T}$ we can compute ${\cal A}(\vec{w})\equiv R^A_{\cal{D}'}(\vec{w})$ in time polynomial in the description of $\vec{w}$ and the description of the distribution ${\cal D}'$. If $A$ is randomized we cannot do this exactly, but we {\em do} get with high probability a good enough approximation for our purposes. We explain how to do this in Appendix~\ref{app:reducedForms}. The outcome is an algorithm ${\cal A}$, which has the following properties with probability at least $1-\eta$, and for arbitrary choices of $\eta \in (0,1)$ and $\gamma \in (0,\alpha)$:
\begin{itemize}
\item for all $\vec{w}$ for which our algorithm from Section~\ref{sec:revenue} may possibly query ${\cal A}$, ${\cal A}$ approximately optimizes the linear objective $\vec{w} \cdot \vec{x}$ over $\vec{x} \in F({\cal F},{\cal D}')$ to within a factor of $(\alpha-\gamma)$;
\item the bit complexity of ${\cal A}(\vec{w})$ is always polynomial in the dimension $T$ and the logarithm of the size, $\poly(n,T,1/\epsilon)$, of the support of ${\cal D}'$;
\item on input $\vec{w}$ of bit complexity $y$, the running time of ${\cal A}$ is 
\begin{align*}&rt_{\cal A}(y) = {\rm poly}(n,T,\hat{\ell},1/\epsilon,\log 1/\eta,1/\gamma, y)\\&~~~~~~~~~~~~~~~~~\cdot rt_A({\rm poly}(n,T,\hat{\ell}, \log 1/\epsilon, y)),\end{align*}
where $rt_A(\cdot)$ represents the running time of $A$ and $\hat{\ell}$ the bit complexity of the coordinates of the points in $\times_i T_i$.
\end{itemize}
Note that replacing $\alpha$ with $\alpha-\gamma$ in Section~\ref{sec:revenue} does not affect our guarantees, except for a loss of a small amount in revenue and truthfulness, which can be made arbitrarily small with $\gamma$.\\

\paragraph{Analyzing the Revenue Optimizing LP.}
First we show that the $WSO$ used in Section~\ref{sec:revenue} as a proxy for a separation oracle for $F(\mathcal{F},\mathcal{D}')$ runs in polynomial time. Recall that the dimension is $d=T$, the bit complexity of ${\cal A}(\vec{w})$ for any $\vec{w}$ can be bounded by $\ell={\rm poly}(n, T, \log 1/\epsilon)$, and that $\gamma$ and $\eta$ are constants used in the definition of ${\cal A}$. Hence,  we immediately get the following corollary of Corollary~\ref{cor:WSOruntime}.

\begin{corollary}
 Let $b$ denote the bit complexity of $\vec{\pi}$. Then on input $\vec{\pi}$, $WSO$ terminates in time {\begin{align*}&\poly(b,n,T, \hat{\ell}, 1/\epsilon,\log 1/\eta,1/\gamma) \\&~~~~~~~~~~~~\cdot rt_A({\rm poly}(n, T, \hat{\ell}, \log 1/\epsilon, b)),\end{align*}} where {$\hat{\ell}$ is an upper bound on the bit complexity of the coordinates of the points in $\times_i T_i$}.
\end{corollary}
Now that we have shown that $WSO$ runs in polynomial time, we need to show that our revenue maximizing LP does as well. The proof of the following is in Appendix~\ref{app:LPruntime}.

\begin{lemma}\label{lem:LPruntime} Let $\hat{\ell}$ denote an upper bound on the bit complexity of $\alpha$, $v_{ij}(B)$ and $\Pr[t_i = B]$ for all $i,j,B$. Then the revenue maximizing LP (if we replace $P$ with $WSO$)\footnote{See what we mean by ``replacing $P$ with $WSO$'' in Section~\ref{sec:revenue}.} terminates in time
{\begin{align*}&\poly(n,T, \hat{\ell}, 1/\epsilon,\log 1/\eta,1/\gamma) \\&~~~~~~~~~~~~\cdot rt_A({\rm poly}(n, T, \hat{\ell}, \log 1/\epsilon)).\end{align*}}
\end{lemma}
With this lemma we complete our proof that our algorithm from Section~\ref{sec:revenue} is both correct and computationally efficient.

\section{Formal Theorem Statements}\label{sec:theorems}

In this section we provide our main theorem, formalizing Informal Theorem~\ref{infthm:appx}. In Appendix~\ref{app:theorems}, we also provide two extensions of our theorem to item-symmetric settings using the techniques of~\cite{DW12}. These extensions are Theorems~\ref{thm:itemsym} and~\ref{thm:bounded} of Appendix~\ref{app:theorems}. In all cases, {\em the allocation rule of the mechanism output by our algorithm is a distribution over virtual implementations of the given social-welfare algorithm $A$.} Moreover, the mechanisms are $\epsilon$-BIC and not truly-BIC, as we only know how to implement the target reduced forms exactly when consumers are sampled from $\mathcal{D}'$ (see discussion in Section~\ref{sec:revenue}). Theorems~\ref{thm:main},~\ref{thm:itemsym} and~\ref{thm:bounded} follow directly from Sections~\ref{sec:WSO} through~\ref{sec:runtime} in the same way that their corresponding theorems (Theorems 6 through 8) in Section~6 of~\cite{CaiDW12b} (arXiv version) follow, after replacing the separation oracle for $F(\mathcal{F},\mathcal{D}')$ with $WSO$ in the LP of Figure~\ref{fig:revenue benchmark}. In all theorem statements, $rt_A(x)$ denotes the runtime of algorithm $A$ on inputs of bit complexity $x$.

\begin{theorem}\label{thm:main} For all $\epsilon, \eta>0$, all $\mathcal{D}$ of finite support in $[0,1]^{nm}$, and all $\mathcal{F}$, given black-box access to a (non-truthful) $\alpha$-approximation algorithm, $A$, for finding the welfare-maximizing allocation in $\mathcal{F}$, there is a polynomial-time randomized approximation algorithm for MDMDP with the following properties: the algorithm obtains expected revenue $\alpha(\opt - \epsilon)$, with probability at least $1-\eta$, in time polynomial in $\ell, n, T,1/\epsilon, \log (1/\eta)$ and $rt_A(\poly(\ell, n,T, \log 1/\epsilon, \log \log (1/\eta)))$, where $\ell$ is an upper bound on the bit complexity of the coordinates of the points in the support of ${\cal D}$, as well as of the probabilities assigned by ${\cal D}_1,\ldots,{\cal D}_m$ to the points in their support. The output mechanism is $\epsilon$-BIC, and can be implemented in the same running time.
\end{theorem}

We remark that we can easily modify Theorem~\ref{thm:main} and its extensions (Theorems~\ref{thm:itemsym} and~\ref{thm:bounded}) to accommodate bidders with hard budget constraints. We simply add into the revenue-maximizing LP constraints of the form $p_i(\vec{v}_i) \leq B_i$, where $B_i$ is bidder $i$'s budget. It is easy to see that this approach works; this is addressed formally in~\cite{CaiDW12,CaiDW12b,DW12}.

\bibliographystyle{plain}
\bibliography{costasbib}

\appendix
\section{Extensions of Theorem~\ref{thm:main}}\label{app:theorems}
This section contains extensions of Theorem~\ref{thm:main} enabled by the techniques of~\cite{DW12}.

\begin{theorem}\label{thm:itemsym} For all $\epsilon, \eta>0$, item-symmetric $\mathcal{D}$ of finite support in $[0,1]^{nm}$, item-symmetric $\mathcal{F}$,\footnote{Distributions and feasibility constraints are item-symmetric if they are invariant under every item permutation.} and given black-box access to a (non-truthful) $\alpha$-approximation algorithm, $A$, for finding the welfare-maximizing allocation in $\mathcal{F}$, there is a polynomial-time randomized approximation algorithm for MDMDP with the following properties: the algorithm obtains expected revenue $\alpha(\opt - \epsilon)$, with probability at least $1-\eta$, in time polynomial in $\ell$, $m,n^c,1/\epsilon, \log (1/\eta)$ and $rt_A(\poly(n^c,m, \log1/\epsilon, \log \log (1/\eta), \ell))$, where $c = \max_{i,j}|\mathcal{D}_{ij}|$, and $|\mathcal{D}_{ij}|$ is the cardinality of the support of the marginal of $\mathcal{D}$ on bidder $i$ and item $j$, and $\ell$ is as in the statement of Theorem~\ref{thm:main}.  The output mechanism is $\epsilon$-BIC, and can be implemented in the same running time.
\end{theorem}

\begin{theorem}\label{thm:bounded} For all $\epsilon, \eta, \delta>0$, item-symmetric $\mathcal{D}$ supported on $[0,1]^{nm}$, item-symmetric $\mathcal{F}$, and given black-box access to a (non-truthful) $\alpha$-approximation algorithm, $A$, for finding the welfare-maximizing allocation in $\mathcal{F}$, there is a polynomial-time randomized approximation algorithm for MDMDP with the following properties: If $C$ is the maximum number of items that are allowed to be allocated simultaneously by $\mathcal{F}$, the algorithm obtains expected revenue $\alpha(\opt - {(\sqrt{\epsilon} + \sqrt{\delta}) C})$, with probability $1-\eta$, in time polynomial in $m, n^{1/\delta}, 1/\epsilon, \log (1/\eta),$ and $rt_{A}(\poly({n^{1/\delta}m},{\log 1/\epsilon}, \log \log 1/\eta))$. In particular, the runtime does \emph{not} depend on $|\mathcal{D}|$ at all). The output mechanism is $\epsilon$-BIC, and can be implemented in the same running time.
\end{theorem}

\begin{remark} The assumption that $\mathcal{D}$ is supported in $[0,1]^{mn}$ as opposed to some other bounded set is w.l.o.g., as we could just scale the values down by a multiplicative $v_{\max}$. This would cause the additive approximation error to be $\epsilon v_{\max}$. In addition, the point of the additive error in the revenue of Theorem~\ref{thm:bounded} is \emph{not} to set $\epsilon,\delta$ so small that they cancel out the factor of $C$, but rather to accept the factor of $C$ as lost revenue. For ``reasonable'' distributions, the optimal revenue scales with $C$, so it is natural to expect that the additive loss should scale with $C$ as well.
\end{remark}

\section{Appendix to Preliminaries}\label{app:prelims}
Here we provide two missing details from the Preliminaries: a formal definition of Bayesian Incentive Compatibility (BIC) and Individual Rationality (IR) for \mattnote{(possibly correlated)} bidders.

\begin{definition}~\cite{DW12}(BIC/$\epsilon$-BIC Mechanism)\label{def:BIC} A direct mechanism $M$ is called $\epsilon$-BIC iff the following inequality holds for all bidders $i$ and types $\tau_i,\tau_i' \in T_i$:
\begin{align*}&\mathbb{E}_{t_{-i} \sim {\cal D}_{-i}{(\tau_i)}}\left[U_i(\tau_i,M_i(\tau_i~;~{t}_{-i}))\right]\\  &~~~~~~~\ge \mathbb{E}_{t_{-i} \sim {\cal D}_{-i}{(\tau_i)}}\left[ U_i(\tau_i,M_i(\tau_i'~;~{t}_{-i})) \right]\\
&~~~~~~~~~~- \epsilon v_{\max} \cdot \max\left\{1,\sum_j \pi^{M}_{ij}(\tau_i',{\tau_i})\right\},\end{align*}
where: 
\begin{itemize}
\item $U_i(B,M_i(C~;~{t}_{-i}))$ denotes the utility of bidder $i$ for the outcome of mechanism $M$ if his true type is $B$, he reports $C$ to the mechanism, and the other bidders report $t_{-i}$;
\item $v_{\max}$ is the maximum possible value of any bidder for any item in the support of the value distribution; and

\item $\pi^{M}_{ij}(A,B)$ is the probability that item $j$ is allocated to bidder $i$ by mechanism $M$ if bidder $i$ reports type $A$ to the mechanism, in expectation over the types of the other bidders, assuming they {are drawn from $\mathcal{D}_{-i}(B)$} and report truthfully, as well as the mechanism's internal randomness.
\end{itemize}
In other words, $M$ is $\epsilon$-BIC iff when a bidder $i$ lies by reporting $\tau_i'$ instead of his true type $\tau_i$, she does not expect to gain more than $\epsilon v_{\max}$ times the maximum of $1$ and the expected number of items that {she would receive by reporting $\tau_i'$ instead}. A mechanism is called BIC iff it is $0$-BIC.\footnote{Strictly speaking, the definition of BIC in~\cite{DW12} is the same but without taking a max with $1$. We are still correct in applying their results with this definition because any mechanism that is considered $\epsilon$-BIC by~\cite{DW12} is certainly considered $\epsilon$-BIC by this definition. We basically call a mechanism $\epsilon$-BIC if either the definition in~\cite{BeiH11,HartlineKM11,HartlineL10} ($\epsilon v_{\max}$) or~\cite{DW12} ($\epsilon v_{\max} \sum_j \pi_{ij}(\tau'_i)$) holds.}
\end{definition}

\begin{definition}\label{def:IR} (IR{/$\epsilon$-IR}) A direct mechanism $M$ is called (interim) $\epsilon$-IR iff the following inequality holds for all bidders $i$ and types $\tau_i \in T_i$:
$$\mathbb{E}_{t_{-i} \sim {\cal D}_{-i}{(\tau_i)}}\left[U_i(\tau_i,M_i(\tau_i~;~{t}_{-i}))\right] \ge -\epsilon,$$
where $U_i(B,M_i(C~;~{t}_{-i}))$ is as in Definition~\ref{def:BIC}. {A mechanism is said to be IR iff it is $0$-IR.}
\end{definition}

%
\section{Input Model} \label{sec:input distribution}
We discuss two models for accessing a value distribution $\mathcal{D}$ over some known $\times_i T_i$, as well as what modifications are necessary, if any, to our algorithms to work with each model:
\begin{itemize}
\item \textbf{Exact Access:} We are given access to a sampling oracle as well as an oracle that exactly integrates the pdf of the distribution over a specified region.
\item \textbf{Sample-Only Access:} We are given access to a sampling oracle and nothing else.
\end{itemize}
The presentation of the paper focuses on the first model. In this case, we can exactly evaluate the probabilities of events without any special care. If we have sample-only access to the distribution, some care is required. Contained in Appendix~A of~\cite{DW12} is a sketch of the modifications necessary for all our results to apply with sample-only access. {The sketch is given for the item-symmetric case, but the same approach will work in the asymmetric case.} Simply put, repeated sampling will yield some distribution $\mathcal{D}'$ that is very close to $\mathcal{D}$ with high probability. If the distributions are close enough, then a solution to the MDMDP for $\mathcal{D}'$ is an approximate solution for $\mathcal{D}$. The error in approximating $\mathcal{D}$ is absorbed into the additive error in both revenue and truthfulness.
\section{Additive Dimension}\label{sec:dimension}
Here we discuss the notion of additive dimension and show some interesting examples of settings with low additive dimension. Consider two settings, both with the same \emph{possible type-space} for each bidder, $\hat{T}_i$ (i.e. $\hat{T}_i$ is the entire set of types that the settings model, $T_i \subseteq \hat{T}_i$ is the set of types that will ever be realized for the given distribution. As a concrete example, $\hat{T}_i = \mathbb{R}^n$ for additive settings.): the first is the ``real'' setting, with the actual items and actual bidder valuations. The real setting has $n$ items, $m$ bidders, feasibility constraints $\mathcal{F}$, and valuation functions $V_{i,B}(S): \mathcal{F} \rightarrow \mathbb{R}$ for all $i,B \in \hat{T}_i$ that map $ S \in \mathcal{F}$ to a value of bidder $i$ of type $B$ for the allocation of items $S$. The second is the ``meta'' setting, with meta-items. The meta-setting has $d$ meta-items, $m$ bidders, feasibility constraints $\mathcal{F}'$, and valuation functions $V'_{i,B}(S'): \mathcal{F}' \rightarrow \mathbb{R}$ for all $i,B \in \hat{T}_i$ that map $S' \in \mathcal{F}'$ to the value of bidder $i$ of type $B$ for the allocation of meta-items $S'$. We now define what it means for a meta-setting to faithfully represent the real setting. 

\begin{definition}\label{def:equiv} A meta-setting is equivalent to a real setting if there is a mapping from $\mathcal{F}$ to $\mathcal{F}'$, $g$, and another from $\mathcal{F}'$ to $\mathcal{F}$, $h$, such that $V_{i,B}(S) = V'_{i,B}(g(S))$, and $V'_{i,B}(S') = V_{i,B}(h(S'))$ for all $i,B \in \hat{T}_i, S \in \mathcal{F}, S' \in \mathcal{F}'$.

\end{definition}

When two settings are equivalent, there is a natural mapping between mechanisms in each setting. Specifically, let $M$ be any mechanism in the real setting. Then in the meta-setting, have $M'$ run $M$, and if $M$ selects allocation $S$ of items, $M'$ selects the allocation $g(S)$ of meta-items and charges exactly the same prices. It is clear that when bidders are sampled from the same distribution, $M$ is BIC/IC/IR if and only if $M'$ is as well. It is also clear that $M$ and $M'$ achieve the same expected revenue. The mapping in the other direction is also obvious, just use $h$. We now define the additive dimension of an auction setting.

\begin{definition} The \emph{additive dimension} of an auction setting is the minimum $d$ such that there is an equivalent (by Definition~\ref{def:equiv}) meta-setting with additive bidders and $d$ meta-items (i.e. due to the feasibility constraints, all bidders valuations can be models as additive over their values for each meta-item).
\end{definition}

In Section~\ref{sec:intro}, we observed that all of our results also apply to settings with additive dimension $d$ after multiplying the runtimes by a $\poly(d)$ factor. This is because a black-box algorithm for approximately maximizing welfare in the real setting is also a black-box algorithm for approximately maximizing welfare in the meta-setting (just apply $g$ to whatever the algorithm outputs). So if we have black-box access to a social welfare algorithm for the real setting, we have black-box access to a social welfare algorithm for the meta-setting. As the meta-setting is additve, all of our techniques apply. We then just apply $h$ at the end and obtain a feasible allocation in the real setting.

We stress that important properties of the setting are not necessarily preserved under the transformation from the real to meta setting. Importantly, when the real setting is downwards closed, this is \emph{not} necessarily true for the meta-setting. The user of this transformation should be careful of issues arising due to negative weights if the desired meta-setting is not downwards-closed.

Respecting the required care, we argued in Section~\ref{sec:intro} that single-minded combinatorial auctions had additive dimension $1$ (and the meta-setting is still downwards-closed, and therefore can accommodate negative values). Now we will show that two other natural models have low additive dimension, and that their corresponding meta-settings are downwards-closed. The discussions below are not intended to be formal proofs. The point of this discussion is to show that interesting non-additive settings have low additive dimension (via meta-settings where approximation algorithms can accommodate negative values) and can be solved using our techniques.

\subsection{$d$-minded Combinatorial Auctions}
A $d$-minded combinatorial auction setting is where each bidder $i$ has at most $d$ (public) subsets of items that they are interested in, and a (private) value $v_{ij}$ for receiving the $j^{th}$ subset in their list, $S_{ij}$,  and value $0$ for receiving any other subset. Such bidders are clearly not additive over their value for the items, but have additive dimension $d$. Specifically, make $d$ meta-items. Define $g(S)$ so that if bidder $i$ receives subset $S_{ij}$ in $S$, they receive item $j$ in $g(S)$. Define $h(S')$ so that if bidder $i$ receives item $j$ in $S'$, they receive the subset of items $S_{ij}$ in $h(S')$. Also define $\mathcal{F}'$ so that an allocation is feasible iff it assigns each bidder at most $1$ meta-item, and when bidder $i$ is assigned meta-item $j_i$, the sets $\{S_{ij_i}| i \in [m]\}$ are pairwise disjoint. Finally, set $V'_{i,B}(j) = V_{i,B}(S_{ij})$. Then it is clear that these two settings are equivalent. It is also clear that bidders are additive in the meta-setting as they are unit-demand (i.e. they can never feasibly receive more than one item). Therefore, $d$-minded Combinatorial Auctions have additive dimension $d$, and any (not necessarily truthful) $\alpha$-approximation algorithm for maximizing welfare implies a (truthful) $(\alpha-\epsilon)$-approximation algorithm for maximizing revenue whose runtime is $\poly(d,T,1/\epsilon,b)$. It is also clear that the meta-setting is downwards-closed, and therefore all (not necessarily truthful) $\alpha$-approximation algorithms for maximizing welfare can accommodate negative values.

\subsection{Combinatorial Auctions with Symmetric Bidders.}
A bidder is symmetric if their value $V_{i,B}(S) = V_{i,B}(U)$ whenever $|S| = |U|$ (i.e. bidders only care about the cardinality of sets they receive). Such bidders (with the extra constraint of submodularity) are studied in~\cite{BadanidiyuruKS12}. Such bidders are again clearly not additive over their values for the items, but have additive dimension $n$. Specifically, make $n$ meta-items. Define $g(S)$ to assign bidder $i$ item $j$ if they received exactly $j$ items in $S$. Define $h(S')$ to assign bidder $i$ exactly $j$ items if they were awarded item $j$ in $S'$ (it doesn't matter in what order the items are handed out, lexicographically works). Also define $\mathcal{F}'$ so that an allocation is feasible iff it assigns each bidder at most $1$ meta-item and when bidder $i$ is assigned meta-item $j_i$, we have $\sum_i j_i \leq n$. Finally, set $V'_{i,B}(j) = V_{i,B}(S)$ where $S$ is any set with cardinality $j$. It is again clear that the two settings are equivalent. It is also clear that the meta-setting is unit-demand, so bidders are again additive. Therefore, combinatorial auctions with symmetric bidders have additive dimension $n$, and any (not necessarily truthful) $\alpha$-approximation algorithm for maximizing welfare implies a (truthful) $(\alpha - \epsilon)$-approximation algorithm for maximizing revenue whose runtime is $\poly(n,T,1/\epsilon,b)$. It is also clear that the meta-setting is downwards-closed, and therefore all (not necessarily truthful) $\alpha$-approximation algorithms for maximizing welfare can accomodate negative values.

In addition, we note here that it is possible to exactly optimize welfare in time $\poly(n,m)$ for symmetric bidders (even with negative, not necessarily submodular values) using a simple dynamic program. We do not describe the algorithm as that is not the focus of this work. We make this note to support that this is another interesting non-additive setting that can be solved using our techniques. \vfill \eject
\section{Omitted Proofs from Section~\ref{sec:WSO}}\label{app:WSO}
\subsection{Omitted Proofs from Section~\ref{subsec:threepolytopes}.}\label{app:threepolytopes}
	
	\begin{prevproof}{Lemma}{lem:chain}
		If $\vec{\pi} \in P_0$, then $\vec{\pi} \cdot \vec{w} \leq \alpha \max_{x\in P} \{x\cdot\vec{w}\} \leq \mathcal{A}(\vec{w})\cdot\vec{w}$ 
		  for all $\vec{w} \in [-1,1]^{d}$, since $\mathcal{A}$ is an $\alpha$-approximation algorithm. 
		  Therefore, $\vec{\pi} \in P_1$ as well. So $P_0 \subseteq P_1$.
		  
		Recall now that a point $\vec{\pi}$ is in the convex hull of $S$ if and only if for all $\vec{w} \in [-1,1]^d$, there exists a point $\vec{x}(\vec{w}) \in S$ such that $\vec{\pi} \cdot \vec{w} \leq \vec{x}(\vec{w}) \cdot \vec{w}$. If $\vec{\pi} \in P_1$, then we may simply let $\vec{x}(\vec{w}) = \mathcal{A}(\vec{w})$ to bear witness that $\vec{\pi} \in P_2(A,\mathcal{D})$.
	\end{prevproof}
	
	\subsection{Omitted Proofs from Section~\ref{subsec:WSO}.}\label{subapp:WSO}
	
		\begin{prevproof}{Fact}{fact:WSOpolytope} Any hyperplane output by $WSO$ is of the form $\vec{w}^* \cdot \vec{\pi} \leq t^*$. Because $\vec{w}^*,t^*$ was accepted by $\widehat{WSO}$, we must have $t^* \geq \mathcal{A}(\vec{w}^*)\cdot\vec{w}^{*}$. As every point in $P_1$ satisfies $\vec{\pi} \cdot \vec{w}^* \leq \mathcal{A}(\vec{w}^*)\cdot\vec{w}^{*} \leq t^*$, we get that $P_1 \subseteq Q$.
	\end{prevproof}
	
	\begin{prevproof}{Fact}{fact:P1}
In order for $WSO$ to reject $\vec{\pi}$, its internal ellipsoid algorithm that uses $\widehat{WSO}$ must find some feasible point $(t^*,\vec{w}^*)$. As $\widehat{WSO}$ accepts $(t^*,\vec{w}^*)$, such a point clearly satisfies the following two equations:

$$t^* < \vec{\pi} \cdot \vec{w}^*$$
$$t^* \geq \mathcal{A}(\vec{w}^*)\cdot\vec{w}^{*}$$

Together, this implies that $\vec{\pi} \cdot \vec{w}^* > \mathcal{A}(\vec{w}^*)\cdot\vec{w}^{*}$, so $\vec{\pi} \notin P_1$.
\end{prevproof}

	\begin{prevproof}{Corollary}{cor:WSOandP1}
	By Fact~\ref{fact:P1}, whenever $WSO$ rejects $\vec{\pi}$, $\vec{\pi} \notin P_{1}$. By Fact~\ref{fact:WSOpolytope}, any halfspace output by $WSO$ contains $P_1$. This is all that is necessary in order for $WSO$ to act as a valid separation oracle for $P_1$ when it rejects $\vec{\pi}$.
	\end{prevproof}
	
		\begin{prevproof}{Lemma}{lem:convhull}
	Define the polytope $P(S)$ as the set of $t,\vec{w}$ that satisfy the following inequalities:
	$$t - \vec{\pi} \cdot \vec{w} \leq - \delta$$
	$$t \geq \mathcal{A}(\vec{w}') \cdot \vec{w},\ \forall \vec{w}' \in S$$
	$$\vec{w} \in [-1,1]^d$$
	
	We first claim that if $\vec{\pi} \notin Conv(\{\mathcal{A}(\vec{w}) | \vec{w} \in S\})$, then $P(S)$ is non-empty. This is because when $\vec{\pi} \notin Conv(\{\mathcal{A}(\vec{w})| \vec{w} \in S\})$, there is some weight vector $\vec{w}^* \in [-1,1]^d$ such that $\vec{w}^* \cdot \vec{\pi} > \max_{\vec{w}' \in S}\{\vec{w}^* \cdot \mathcal{A}(\vec{w}')\}$. For appropriately chosen $\delta$ (Lemma~\ref{lem:delta} in Section~\ref{sec:runtime} provides one), there is also a $\vec{w}^*$ such that $\vec{w}^* \cdot \vec{\pi} \geq \max_{\vec{w}' \in S}\{\vec{w}^* \cdot \mathcal{A}(\vec{w}')\} + \delta$. Set $t^* := \max_{\vec{w}' \in S}\{\vec{w}^* \cdot \mathcal{A}(\vec{w}')\}$ and consider the point $(t^*,\vec{w}^*)$. As $t^*$ is larger than $\mathcal{A}(\vec{w}')\cdot \vec{w}^*$ for all $\vec{w}' \in S$ by definition, and $\vec{w}^* \in [-1,1]^d$ by definition, we have found a point in $P(S)$.
	
	Now, consider that in the execution of $WSO(\vec{\pi})$, $\widehat{WSO}$ outputs several halfspaces. As $S$ is exactly the set of weights $\vec{w}$ that $WSO$ queries to $\widehat{WSO}$, these are exactly the halfspaces:
	
	$$t \geq \mathcal{A}(\vec{w}')\cdot \vec{w}, \ \forall \vec{w}' \in S$$
	
	During the execution of $WSO(\vec{\pi})$, other halfspaces may be learned not from $\widehat{WSO}$, but of the form $t-\vec{\pi} \cdot \vec{w} \leq -\delta$ or $-1 \leq w_{i}$, $w_i \leq 1$ for some $i\in[d]$. Call the polytope defined by the intersection of all these halfspaces $P(WSO)$. As all of these halfspaces contain $P(S)$, it is clear that $P(WSO)$ contains $P(S)$. 
	
	Now we just need to argue that if $N$ is sufficiently large, and $WSO(\vec{\pi})$ could not find a feasible point in $N$ iterations, then $P(S)$ is empty. Corollary~\ref{cor:N} in Section~\ref{sec:runtime} provides an appropriate choice of $N$. Basically, if $P(S)$ is non-empty, we can lower bound its volume with some value $V$ (independent of $S$). If $N = \poly(\log (1/V))$, then the volume of the ellipsoid containing $P(WSO)$ after $N$ iterations will be strictly smaller than $V$. As $P(WSO)$ contains $P(S)$, this implies that $P(S)$ is empty. Therefore, we may conclude that $\vec{\pi} \in Conv(\{\mathcal{A}(\vec{w}) | \vec{w} \in S\})$.
	\end{prevproof}
\section{Omitted Figures and Proofs from Section~\ref{sec:revenue}}\label{app:revenue}
	\begin{figure}[ht]
	\colorbox{MyGray}{
	\begin{minipage}{0.46\textwidth} {
	\noindent\textbf{Near-Optimal LP:}\\
	
	\noindent\textbf{Variables:}
	\begin{itemize}
	\item $p_i(\vec{v}_i)$, for all bidders $i$ and types $\vec{v}_i \in T_i$, denoting the expected price paid by bidder $i$ when reporting type $\vec{v}_i$ over the randomness of the mechanism and the other bidders' types.
	\item $\pi_{ij}(\vec{v}_i)$, for all bidders $i$, items $j$, and types $\vec{v}_i \in T_i$, denoting the probability that bidder $i$ receives item $j$ when reporting type $\vec{v}_i$ over the randomness of the mechanism and the other bidders' types.
	\end{itemize}
	\textbf{Constraints:}
	\begin{itemize}
	\item $\vec{\pi}_i(\vec{v}_i) \cdot \vec{v}_i - p_i(\vec{v}_i) \geq \vec{\pi}_i(\vec{w}_i)\cdot \vec{v}_i - p_i(\vec{w}_i)$, for all bidders $i$, and types $\vec{v}_i,\vec{w}_i \in T_i$, guaranteeing that the reduced form mechanism $(\vec{\pi},\vec{p})$ is BIC.
	\item $\vec{\pi}_i(\vec{v}_i) \cdot \vec{v}_i - p_i(\vec{v}_i) \geq 0$, for all bidders $i$, and types $\vec{v}_i \in T_i$, guaranteeing that the reduced form mechanism $(\vec{\pi},\vec{p})$ is individually rational.
	\item $\vec{\pi}\in F(\mathcal{F},\mathcal{D}')$, guaranteeing that the reduced form $\vec{\pi}$ is feasible for $\mathcal{D}'$.
	\end{itemize}
	\textbf{Maximizing:}
	\begin{itemize}
	\item $\sum_{i=1}^{m} \sum_{\vec{v}_i \in T_i} \Pr[t_i = \vec{v}_i]\cdot p_i(\vec{v}_i)$, the expected revenue \textbf{when played by bidders sampled from the true distribution $\mathcal{D}$}.\\
	\end{itemize}}
	\end{minipage}}
	\caption{An LP with near-optimal revenue.}
	\label{fig:revenue benchmark}
	\end{figure}
	
	\begin{prevproof}{Lemma}{lem:goodRevenue}
Let $(\vec{\pi}^*,\vec{p}^*)$ denote the reduced form output by the LP in Figure~\ref{fig:revenue benchmark}. Then we claim that the reduced form $(\alpha \vec{\pi}^*,\alpha \vec{p}^*)$ is a feasible solution after replacing $F(\mathcal{F},\mathcal{D}')$ with $P_0$. It is clear that this mechanism is still IR and BIC, as we have simply multiplied both sides of every incentive constraint by $\alpha$. It is also clear that ${\alpha}\vec{\pi}^* \in P_0$ by definition. As we have multiplied all of the payments by $\alpha$, and the original LP had revenue $\opt - \epsilon$~\cite{CaiDW12b}, it is clear that revenue of the reduced form output in the new LP is at least $\alpha (\opt - \epsilon)$.
\end{prevproof}

\begin{prevproof}{Proposition}{prop:WSOrevenue}
Let $Q_{0}$ be the set of $(\vec{\pi},\vec{p})$ that satisfy $\vec{\pi} \in P_0$, the BIC and IR constraints, as well as the revenue constraint $\sum_{i=1}^{m} \sum_{\vec{v}_i \in T_i} \Pr[t_i = \vec{v}_i]\cdot p_i(\vec{v}_i) \geq x$. As $x \leq Rev(P_0)$, we know that there is some feasible point $(\vec{\pi}^*,\vec{p}^*)\in Q_{0}$. Therefore, the ellipsoid algorithm using a valid separation oracle for $Q_0$ and the correct parameters will find a feasible point. 

Now, what is the difference between a valid separtion oracle for $Q_0$ and $WSO'$ as used in Proposition~\ref{prop:WSOrevenue}? A separation oracle for $Q_0$ first checks the BIC, IR, and revenue constraints, then executes a true separation oracle for $P_0$. $WSO'$ first checks the BIC, IR, and revenue constraints, then executes $WSO$. So let us assume for contradiction that the Ellipsoid using $WSO'$ outputs infeasible, but $Q_0$ is non-empty. 
It has to be then that $WSO$ rejected every point that was queried to it.\footnote{In particular, even if Ellipsoid deems the feasible region lower-dimensional, and continues in a lower-dimensional space, etc., then still if the final output of Ellipsoid is infeasible, then all points that it queried to $WSO$ were rejected.} However, Corollary~\ref{cor:WSOandP1} guarantees that when rejecting points, $WSO$ acts as a valid separation oracle for $P_0$ (i.e. provides a valid hyperplane separating $\vec{\pi}$ from $P_0$). As the only difference between a separation oracle for $Q_0$ and $WSO'$ was the use of $WSO$, and $WSO$ acted as a valid separation oracle for $P_0$, this means that in fact $WSO'$ behaved as a valid separation oracle for $Q_0$. So we ran the ellipsoid algorithm using a valid separation oracle for $Q_0$ with the correct parameters, but output ``infeasible'' when $Q_0$ was non-empty, contradicting the correctness of the ellipsoid algorithm.

Therefore, whenever $Q_0$ is non-empty, $WSO'$ must find a feasible point. As $Q_0$ is non-empty whenever $x \leq Rev(P_0)$, this means that $WSO'$ will find a feasible point whenever $x \leq Rev(P_0)$, proving the proposition.\end{prevproof}

\begin{prevproof}{Corollary}{cor:WSOrevenue}
Consider running the LP of Figure~\ref{fig:revenue benchmark} with $WSO$. The optimization version of the ellipsoid algorithm will do a binary search on possible values for the objective function and solve a separate feasibility subproblem for each. Proposition~\ref{prop:WSOrevenue} guarantees that on every feasibility subproblem with $x \leq Rev(P_0)$, the ellipsoid algorithm will find a feasible point. Therefore, the binary search will stop at some value $x^* \geq Rev(P_0)$, and we get that $Rev(WSO) \geq Rev(P_0)$.
\end{prevproof}

\begin{prevproof}{Corollary}{cor:implement}
Because $\vec{\pi}^*$ is output, we have $WSO(\vec{\pi}^*) =$ ``yes.'' Lemma~\ref{lem:convhull} tells us that $\vec{\pi}^*$ is therefore in the convex hull of $\{R^A_{\mathcal{D}'}(\vec{w})| \vec{w} \in S\}$. As a convex combination of reduced forms can be implemented as a distribution over the allocation rules that implement them, we have proven the corollary.
\end{prevproof}

\subsection{The Solution is $\epsilon$-BIC.}\label{app:epsBIC}

Here we provide the omitted informal argument from Section~\ref{sec:revenue}. Again, the reader is referred to~\cite{CaiDW12b} for a formal proof. Recall that we are trying to show that the distribution over virtual implementations of $A$ used to implement $\vec{\pi}^*$ when bidders are sampled from $\mathcal{D}'$ implements some reduced form $\vec{\pi}'$ when bidders are sampled from $\mathcal{D}$ satisfying $|\vec{\pi}^*-\vec{\pi}'|_1 \leq \epsilon$. This suffices to guarantee that our mechanism computed in Section~\ref{sec:revenue} is $\epsilon$-BIC.

Let $x$ denote the number of samples taken for $\mathcal{D}'$. There are two steps in the argument: the first is showing that for a fixed allocation rule, the probability that its reduced form when bidders are sampled from $\mathcal{D}'$ is within $\epsilon$ in $\ell_1$ distance of the reduced form when bidders are sampled from $\mathcal{D}$ approaches $1$ exponentially fast in $x$. This can be done by a simple Chernoff bound. The second is showing that, before actually sampling $\mathcal{D}'$, but after choosing how many samples to take, the number of weights $\vec{w}$ that could possibly ever be queried to $\widehat{WSO}$ grows like $2^{\poly(\log x)}$. This can be done by reasoning about bit complexities maintained by the ellipsoid algorithm, and the fact that the bit complexity of $\mathcal{D}'$ grows logarithmically in $x$. With both facts, we can then take a union bound over the entire set of weights that will ever be possibly queried to $\widehat{WSO}$ and get that with high probability, the reduced forms of all the allocation rules corresponding to these weights are $\epsilon$-close under $\mathcal{D}$ and $\mathcal{D}'$. If this holds for all these allocation rules simultaneously, then it clearly also holds for any distribution over them (in particular, for $\vec{\pi}^*$). This entire argument appears formally in~\cite{CaiDW12b} for the case where the LP of Figure~\ref{fig:revenue benchmark} uses $F(\mathcal{F},\mathcal{D}')$ rather than $WSO$, and is nearly identical.
\section{Omitted Proofs from Section~\ref{sec:runtime}}\label{app:runtime}

\subsection{Omitted Proofs on the Runtime of $WSO$ from Section~\ref{sec:runtime}.}\label{app:WSOruntime}

\begin{prevproof}{Lemma}{lem:delta}
Consider the polytope $P'(S)$ with respect to variables $t'$ and $\vec{w}'$ that is the intersection of the following half-spaces (similar to $P(S)$ from the proof of Lemma~\ref{lem:convhull}):
	$$t' \leq d$$
	$$t' \geq \mathcal{A}(\vec{w}) \cdot \vec{w}',\ \forall \vec{w} \in S$$
	$$\vec{w}' \in [-1,1]^d$$
	
	If $\vec{\pi} \notin Conv(\{\mathcal{A}(\vec{w}) | \vec{w}\in S\})$, there exists some weight vector $\vec{w}'$ such that $\vec{\pi} \cdot \vec{w}' > \max_{\vec{w} \in S} \{\mathcal{A}(\vec{w})\cdot \vec{w}'\}$. This bears witness that there is a point in $P'(S)$ satisfying $\vec{\pi} \cdot \vec{w}' > t'$. If such a point exists, then clearly we may take $(\vec{w}',t')$ to be a corner of $P'(S)$ satisfying the same inequality. As the bit complexity of every halfspace defining $P'(S)$ is $\poly(d,\ell)$, the corner also has bit complexity $\poly(d,\ell)$. Therefore, if $t' - \vec{\pi} \cdot \vec{w}' < 0$, $t' - \vec{\pi} \cdot \vec{w}' \leq -{4}\delta$, for some $\delta = 2^{-\poly(d,\ell,b)}$. 
\end{prevproof}

\begin{lemma}\label{lem:goodvolume} Let $S$ be any subset of weights. Let also $b$ be the bit complexity of $\vec{\pi}$, and {$\ell$ an upper bound on the bit complexity of ${\cal A}(\vec{w})$ for all $\vec{w}$.} Then, if we choose $\delta$ as prescribed by Lemma~\ref{lem:delta}, the following polytope ($P(S)$) is either empty, or has volume at least $2^{-\poly(d,\ell,b)}$:

	$$t - \vec{\pi} \cdot \vec{w} \leq - \delta$$
	$$t \geq \mathcal{A}(\vec{w}') \cdot \vec{w},\ \forall \vec{w}' \in S$$
	$$\vec{w} \in [-1,1]^d$$
\end{lemma}

\begin{prevproof}{Lemma}{lem:goodvolume}
First, it will be convenient to add the vacuous constraint $t \leq d$ to the definition of the polytope. It is vacuous because it is implied by the existing constraints,{\footnote{{Since $P\in[-1,1]^{d}$, WLOG we can also assume $\vec{\pi}\in [-1,1]^{d}$, thus from the constraints of $P(S)$ it follows that $t\leq d$.}}} but useful for the analysis. Define $P'(S)$ by removing the first constraint. That is, $P'(S)$ is the intersection of the following halfspaces (this is the same as $P'(S)$ from the proof of Lemma~\ref{lem:delta}):

	$$t \leq d$$
	$$t \geq \mathcal{A}(\vec{w}') \cdot \vec{w},\ \forall \vec{w}' \in S$$
	$$\vec{w} \in [-1,1]^d$$
	
	If there is a point in $P(S)$, then there is some point in $P'(S)$ satisfying $t - \vec{\pi} \cdot \vec{w} \leq -\delta$. If such a point exists, then clearly there is also a corner of $P'(S)$ satisfying $t - \vec{\pi} \cdot \vec{w} \leq - \delta$. Call this corner $(t^*,\vec{w}^*)$. Recall that $\delta$ was chosen in the proof of Lemma~\ref{lem:delta} so that we are actually guaranteed $t - \vec{\pi} \cdot \vec{w} \leq -4\delta$. Therefore, the point $(t^*/2,\vec{w}^*/2)$ is also clearly in $P(S)$, and satisfies $t-\vec{\pi} \cdot \vec{w} \leq -2\delta$. 
	
	Now, consider the box $B =[\frac{t^*}{2}+\frac{\delta}{2},\frac{t^*}{2}+\frac{3\delta}{4}] \times (\times_{i=1}^d [\frac{w^*_i}{2},\frac{w^*_i}{2}+\frac{\delta}{4d}])$. We claim that $B \subseteq P(S)$. Let $(t,\vec{w})$ denote an arbitrary point in $B$. It is clear that we have $\vec{w} \in [-1,1]^d$, as we had $\vec{w}^*/2 \in [-1/2,1/2]^d$ to start with. As each coordinate of $\vec{\pi}$ and $\mathcal{A}(\vec{w}')$ for all $\vec{w}'$ is in $[-1,1]$, it is easy to see that:
	
	$$(\vec{w}^*/2) \cdot \vec{\pi}-\frac{\delta}{4} \leq \vec{w} \cdot \vec{\pi}\leq  (\vec{w}^*/2)\cdot \vec{\pi}+\frac{\delta}{4}, $$ and for all $\vec{w}' \in S$,
	$$(\vec{w}^*/2) \cdot \mathcal{A}(\vec{w}')-\frac{\delta}{4} \leq \vec{w} \cdot \mathcal{A}(\vec{w}') \leq (\vec{w}^*/2)\cdot \mathcal{A}(\vec{w}') +\frac{\delta}{4}.$$
	
	As we must have $t \geq \frac{t^*}{2} + \frac{\delta}{2}$, and we started with $t^* \geq \vec{w}^* \cdot \mathcal{A}(\vec{w}')$ for all $\vec{w}' \in S$, it is clear that we still have $t \geq \vec{w} \cdot \mathcal{A}(\vec{w}')$ for all $\vec{w}' \in S$. Finally, since we started with $t^*/2 - \vec{\pi} \cdot \vec{w}^*/2 \leq -2\delta$, and $t \leq t^*/2 + \frac{3\delta}{4}$, we still have $t - \vec{\pi}\cdot \vec{w} \leq -\delta$.
	
	Now, we simply observe that the volume of $B$ is $\frac{\delta^{d+1}}{d^d4^{d+1}}$, which is $2^{-\poly(d,\ell,b)}$. Therefore, if $P(S)$ is non-empty, it contains this box $B$, and therefore has volume at least $2^{-\poly(d,\ell,b)}$.
\end{prevproof}

\begin{prevproof}{Corollary}{cor:N}
By Lemma~\ref{lem:goodvolume}, if $P(S)$ is non-empty, $P(S)$ has volume at least some $V=2^{-\poly(d,\ell,b)}$. Since $P \subseteq [-1,1]^d$, the starting ellipsoid in the execution of WSO can be taken to have volume {$2^{O(d)}$}. As the volume of the maintained ellipsoid shrinks by a multiplicative factor of at least $1-\frac{1}{\poly(d)}$ in every iteration of the ellipsoid algorithm, after some $N=\poly(d,\ell,b)$ iterations, we will have an ellipsoid with volume smaller than $V$ that contains $P(S)$ (by the proof of Lemma~\ref{lem:convhull}), a contradiction. Hence $P(S)$ must be empty, if we use the $N$ chosen above for the definition of $WSO$ and the ellipsoid algorithm in the execution of $WSO$ does not find a feasible point after $N$ iterations.
\end{prevproof}

\begin{prevproof}{Corollary}{cor:WSOruntime}
By the choice of $N$ in Corollary~\ref{cor:N}, $WSO$ only does $\poly(d,\ell,b)$ iterations of the ellipsoid algorithm. {Note that the starting ellipsoid can be taken to be the sphere of radius $\sqrt{d}$ centered at $\vec{0}$, as $P \subseteq [-1,1]^d$. Moreover, the hyperplanes output by $\widehat{WSO}$ have bit complexity $O(\ell)$, while all other hyperplanes that may be used by the ellipsoid algorithm have bit complexity ${\rm poly}(d,\ell,b)$ given our choice of $\delta$. So by \cite{GLS1988}, $\widehat{WSO}$ will only be queried at points of bit complexity ${\rm poly}(d,\ell,b)$, and every such query will take time $\poly(\poly(d,\ell,b), rt_{\mathcal{A}}(\poly(d,\ell,b))$ as it involves checking one inequality for numbers of bit complexity $\poly(d,\ell,b)$ and making one call to ${\cal A}$ on numbers of bit complexity $\poly(d,\ell,b)$.} 
Therefore, $WSO$ terminates in the promised running time.
\end{prevproof}

\subsection{Computing Reduced Forms of (Randomized) Allocation Rules.}\label{app:reducedForms}
For distributions that are explicitly represented (such as $\mathcal{D}'$), it is easy to compute the reduced form of a deterministic allocation rule: simply iterate over every profile in the support of ${\cal D}'$, run the allocation rule, and see who receives what items. For randomized allocation rules, this is trickier as computing the reduced form exactly would require enumerating over the randomness of the allocation rule. One approach is to approximate the reduced form. This approach works, but is messy to verify formally, due to the fact that the bit complexity of reduced forms of randomized allocations takes effort to bound. The technically cleanest approach is to get our hands on a deterministic allocation rule instead.

Let $A$ be a randomized allocation rule that obtains an $\alpha$-fraction of the maximum welfare in expectation. Because the welfare of the allocation output by $A$ cannot be larger than the maximum welfare, the probability that $A$ obtains less than an $(\alpha-\gamma)$-fraction of the maximum welfare is at most $1-\gamma$. So let $A'$ be the allocation rule that runs several independent trials of $A$ and chooses (of the allocations output in each trial) the one with maximum welfare. If the number of trials is $x/\gamma$, we can guarantee that $A'$ obtains at least an $(\alpha-\gamma)$-fraction of the maximum welfare with probability $1-e^{-x}$. From this it follows that, if $O((\ell+\tau)/\gamma)$ independent trials of $A$ are used for $A'$, then $A'$ obtains an $(\alpha-\gamma)$-fraction of the maximum welfare for \emph{all} input vectors of bit complexity $\ell$, with probability at least $1-2^{-\tau}$. This follows by taking a union bound over all $2^\ell$ possible vectors of bit complexity $\ell$. For $\ell, \tau$ to be determined later, we fix the randomness used by $A'$ in running $A$ ahead of time so that $A'$ is a deterministic algorithm. Define $\mathcal{A}'$ using $A'$ in the same way that $\mathcal{A}$ is defined using $A$. 

As $A'$ is a deterministic algorithm and $\mathcal{D}'$ a uniform distribution over ${\rm poly}(n, T,1/\epsilon)$ profiles, we can compute $R^{\mathcal{A}'}_{\mathcal{D}'}(\vec{w})$ for a given $\vec{w}$ by enumerating over the support of ${\cal D}'$ as described above. The resulting $R^{\mathcal{A}'}_{\mathcal{D}'}(\vec{w})$ has bit complexity polynomial in the dimension $T$ and the logarithm of $\poly(n,T,1/\epsilon)$.\footnote{{This is because the probability that bidder $i$ gets item $j$ conditioned on being type $B$ is just the number of profiles in the support of $\mathcal{D}'$ where $t_i = B$ and bidder $i$ receives item $j$ divided by the number of profiles where $t_i = B$. The definition of ${\cal D}'$ (see~\cite{CaiDW12b}) makes sure that the latter is non-zero.}}

Now let us address the choice of $\ell$. We basically want to guarantee the following. Suppose that we use $\mathcal{A}'$ inside $WSO$. We want to guarantee that $\mathcal{A}'$ will work well for any vector $\vec{w}$ that our algorithm will possibly ask ${\cal A}'$.\footnote{Namely, we want that $\mathcal{A}'$ approximately optimizes the linear objective $\vec{w}\cdot \vec{x}$ over $\vec{x} \in F({\cal F},{\cal D}')$ to within a multiplicative factor $\alpha-\gamma$.} We argue in the proof of Lemma~\ref{lem:LPruntime} that, regardless of the bit complexity of the hyperplanes output by $WSO$, throughout the execution of our algorithm $WSO$ will only be queried on points of bit complexity $\poly(n,T,\hat{\ell},\log 1/\epsilon)$, where $\hat{\ell}$ is the bit complexity of the coordinates of the points in $\times_i T_i$. From Corollary~\ref{cor:WSOruntime} it follows then that ${\cal A}'$ will only be queried on inputs of bit complexity ${\rm poly}(n,T,\hat{\ell},\log 1/\epsilon)$. This in turns implies that $A'$ will only be queried on inputs of bit complexity ${\rm poly}(n,T,\hat{\ell},\log 1/\epsilon)$. So setting $\ell$ to some ${\rm poly}(n,T,\hat{\ell},\log 1/\epsilon)$ guarantees that $A'$ achieves an $(\alpha-\gamma)$-fraction of the maximum welfare, for all possible inputs it may be queried simultaneously, with probability at least $1-2^{-\tau}$.

Finally, for every input $\vec{w}$ of bit complexity $x$, the running time of ${\cal A}'$ is polynomial in $x$, the support size and the bit complexity of ${\cal D}'$ (which is ${\rm poly}(n,T,\hat{\ell},1/\epsilon)$), and the running time of $A'$ on inputs of bit complexity ${\rm poly}(n,T,\hat{\ell}, \log 1/\epsilon, x)$. The latter is just a factor of ${\rm poly}(n,T,\hat{\ell},\log 1/\epsilon,\tau,1/\gamma)$ larger than that of $A$ on inputs of bit complexity ${\rm poly}(n,T,\hat{\ell}, \log 1/\epsilon, x)$. Overall, 
\begin{align*}
&rt_{{\cal A}'}(x)= {\rm poly}(n,T,\hat{\ell},1/\epsilon,\tau,1/\gamma, x) \\&~~~~~~~~~~~~~~~~~~~\cdot rt_A({\rm poly}(n,T,\hat{\ell}, \log 1/\epsilon, x)).
\end{align*}

%
%
%
%
%

\subsection{Omitted Proofs on the Runtime of the Revenue-Maximizing LP from Section~\ref{sec:runtime}.}\label{app:LPruntime}
\begin{prevproof}{Lemma}{lem:LPruntime}
Ignoring computational efficiency, we can use the construction of~\cite{CaiDW12b} to build a separation oracle for $P_0$, defined as in Section~\ref{sec:revenue}. Suppose that we built this separation oracle, $SO$, and used it to solve the LP of Figure~\ref{fig:revenue benchmark} with $P_0$ in place of $F(\mathcal{F},\mathcal{D}')$ using the ellipsoid algorithm. It follows from~\cite{GLS1988} that the ellipsoid algorithm using $SO$ would terminate in time polynomial in $n,T,\hat{\ell},\log 1/\epsilon$ and the running time of $SO$ on points of bit complexity $\poly(n,T,\hat{\ell},\log 1/\epsilon)$.\footnote{Note that for any guess $x$ on the revenue, we can upper bound the volume of the resulting polytope by $2^{O(T)}$ and lower bound it by some $2^{- {\rm poly}(n,T,\hat{\ell},\log 1/\epsilon)}$, whatever its dimension is. We can also take the precision to be ${\rm poly}(n,T,\hat{\ell}, \log 1/\epsilon)$.} As we are running exactly this algorithm (i.e. with the same parameters and criterion for deeming the feasible region lower-dimensional), except replacing the separation oracle for $P_0$ with $WSO$, our solution will also terminate in time polynomial in $n,T,\hat{\ell},\log 1/\epsilon$ and the runtime of $WSO$ on points of bit complexity $\poly(n,T,\hat{\ell}, \log1/\epsilon)$. Indeed, for every guess on the revenue and as long as the Ellipsoid algorithm for that guess has not terminated, it must be that $WSO$ has been rejecting the points that it has been queried, and by Corollary~\ref{cor:WSOandP1} in this case it acts as a valid separation oracle for $P_0$, and hence is input points of the same bit complexity that could have been input to $SO$, namely $\poly(n,T,\hat{\ell},\log 1/\epsilon)$. Corollary~\ref{cor:WSOruntime} shows that the runtime of $WSO$ on points of bit complexity $\poly(n,T,\hat{\ell},\log 1/\epsilon)$ is 
\begin{align*}&\poly(n,T, \hat{\ell}, 1/\epsilon,\log 1/\eta,1/\gamma) \\&~~~~~~~~~~~~\cdot rt_A({\rm poly}(n, T, \hat{\ell}, \log 1/\epsilon)),\end{align*}
so the entire running time of our algorithm is as promised.
\end{prevproof}
\section{Generic Theorem Statement}

\yangnote{
The following theorem summarizes our results for approximately optimizing over a convex polytope using a weird separation oracle. The summary is stated in a more generic form than is required for the main results of this paper so that it can be applied more easily in future work.
\begin{theorem}
Let $P$ be a $d$-dimensional bounded polytope containing the origin, and let $\mathcal{A}$ be any algorithm that takes any direction $\vec{w}\in[-1,1]^{d}$ as input and outputs a point $\mathcal{A}(\vec{w}) \in P$ such that $\mathcal{A}(\vec{w})\cdot \vec{w}\geq \alpha\cdot \max_{\vec{x}\in P} \vec{x}\cdot\vec{w}$ for some absolute constant $\alpha\in(0,1]$. Then we can design a ``weird'' separation oracle $WSO$ such that,
\begin{enumerate}
\item Every halfspace output by $WSO$ will contain  $\alpha P=\{\alpha\cdot\vec{\pi}\ |\ \vec{\pi} \in P\}$.
\item Whenever $WSO(\vec{x}) =$ ``yes,'' the execution of $WSO$ explicitly finds directions $\vec{w}_1,\ldots,\vec{w}_k$ such that $\vec{x} \in \text{Conv}\{\mathcal{A}(\vec{w}_1),\ldots, \mathcal{A}(\vec{w}_k)\}$.
\item Let $b$ be the bit complexity of $\vec{x}$, $\ell$ be an upper bound of the bit complexity of $\mathcal{A}(\vec{w})$ for all $\vec{w}\in[-1,1]^{d}$, and $rt_{\mathcal{A}}(y)$ be the running time of algorithm $\mathcal{A}$ on some input with bit complexity $y$. Then on input $\vec{x}$, $WSO$ terminates in time $\poly(d,b,\ell,rt_{\mathcal{A}}(\poly(d,b,\ell)))$ and makes at most $\poly(d,b,\ell)$ many queries to $\mathcal{A}$.
\end{enumerate} \end{theorem}
\begin{proof}
In fact, we have already proved all three claims in some previous Lemmas and Corollaries. For each claim, we point out where the proof is. For the first claim, using Corollary~\ref{cor:WSOandP1} we know $P_{1}$ is contained in all halfspaces output by $WSO$, therefore $\alpha P\subseteq P_{1}$ is also contained in all halfspaces. The second claim is proved in Lemma~\ref{lem:convhull}, and the third claim is proved in Corollary~\ref{cor:WSOruntime}.
\end{proof}
}
\end{document}